\documentclass[preprint]{elsarticle}
\usepackage{amsgen,amsmath,amstext,amsbsy,amsopn,amsfonts,amssymb, stmaryrd,amsthm}
\usepackage{tcolorbox}
\usepackage[english]{babel}
\usepackage{tikz}
\usetikzlibrary{chains,positioning,scopes,shapes} 
\usepackage{color}
\usepackage{caption}
\usepackage{subcaption}

\usepackage{algorithm,xspace}
\usepackage[noend]{algorithmic}

\newtheorem{theorem}{Theorem}
\newtheorem{proposition}{Proposition}

\newcommand{\remove}[1]{}
\newcommand{\gn}[1]{\texttt{\tt #1}}

\newcommand{\vp}{\varphi}

\usepackage{xspace}

\newcommand{\MGRP}{\textsc{MGR}\xspace}
\newcommand{\MsGRP}{\textsc{MsGR}\xspace}
\newcommand{\CR}{\textsc{CR}\xspace}
\begin{document}
\begin{frontmatter}
  \title{The  Multicolored Graph Realization Problem}

\author[1,4]{Josep D\'{\i}az}\ead{diaz@cs.upc.edu} 
\author[2,3]{ \"{O}znur Ya\c{s}ar Diner}\ead{oznur.yasar@khas.edu.tr}
\author[1,4]{Maria Serna\corref{cor1}}\ead{mjserna@cs.upc.edu}
\author[3,4]{Oriol Serra}\ead{oriol.serra@upc.edu}

\affiliation[1]{organization = {ALBCOM Research Group, Computer Science Department.  Universitat Polit\`{e}cnica de Catalunya}, city = {Barcelona}, country = {Spain}}
\affiliation[2]{organization = {Computer Engineering Department, Kadir Has University}, city = {Istanbul}, , country = {Turkey}}
\affiliation[3]{organization = {Mathematics Department.  Universitat Polit\`{e}cnica de Catalunya}, city = {Barcelona}, country = {Spain}}
\affiliation[4]{organization = {Institut de Matem\`{a}tiques de la UPC-BarcelonaTech (IMTech), Universitat Polit\`{e}cnica de Catalunya},  city = {Barcelona}, country = {Spain}}

\cortext[cor1]{Corresponding author}

\begin{abstract}
We introduce  the \emph{multicolored graph realization} problem (\MGRP). The input to the problem is  a colored graph $(G,\vp)$, i.e., a graph   together with a coloring $\varphi$  on its vertices. We can associate to each colored graph $(G,\vp)$ a \emph{cluster graph}  ($G_\varphi$)  in which, after collapsing to a node all vertices with the same color, we remove multiple edges and self-loops. A set of vertices $S$ is \emph{multicolored} when $S$ has exactly one vertex from each color class.  The problem is  to decide whether there is a multicolored set $S$ such that, after identifying each vertex in $S$ with its color class,   $G[S]$ coincides with  $G_\vp$.
  
The \MGRP  problem is related  to the well known class of generalized network problems, most of which are NP-hard.  For example  the generalized Minimun Spanning Tree problem. \remove{It generalizes the  multicolored clique problem, introduced in the context of parameterized complexity,  which is known to be  $W[1]$-hard, when  parameterized by the number of colors. We are interested in analyzing the parameterized complexity of the problem with respect to different parameterizations.  
}
The  \MGRP problem  is a generalization of the  multicolored clique problem, which is known to be $W[1]$-hard when parameterized by the number of colors. Thus \MGRP  remains $W[1]$-hard, when parameterized by the size of the cluster graph.  This results implies that the \MGRP problem is $W[1]$-hard when parameterized by any graph parameter on $G_\vp$, among those  for treewidth.  In consequence, we look to instances of the problem in which both the number of color classes and the treewidth of $G_\vp$ are unbounded.  We consider three natural such  graph classes: chordal graphs, convex bipartite graphs and 2-dimensional grid graphs.  We show that the \MGRP problem is NP-complete when $G_\vp$ is  either chordal,  biconvex bipartite, complete bipartite or a 2-dimensional grid.   Our hardness results follows from suitable reductions from the 1-in-3 monotone SAT problem. Our  reductions show that the problem remains hard even when the maximum number of vertices in a color class is 3. In the case of the grid, the hardness holds also graphs with  bounded degree. We complement those results by showing combined parameterizations under which  the \MGRP problem  became tractable.  
 \end{abstract}

\begin{keyword}
multicolored  realization problem\sep generalized combinatorial problems \sep parameterized complexity \sep convex bipartite graphs
\end{keyword}

\end{frontmatter}

\section{Introduction}

It is well known that graphs are an important tool to model many systems in different disciplines. In particular
graph partitioning and graph clustering are key techniques  in various areas of Computer Science, Engineering, Biology, Epidemiology, Social Science, etc. 
 For example when dealing with the analysis of large social nets, the modelization of infection spreading,   route planning, community detection in social
networks and high performance computing. In many of these applications large graphs
are partitioned as to control the structural connections among the clusters (the elements of
the partition).  Given a partition of the vertices into clusters,  the   topological notion of  the \emph{graph quotient} provides a way to obtain a cluster graph as a summary of the input graph.  The cluster graph  provides a simpler and compact form of  complex network, extracted from an adequate partition of the data, summarizing the relevant relationships among the clustered data.    Graph quotients have many  applications in the study of data sets containing complex relationships (see for example \cite{Vepstas17}) and have motivated the study  of generalized  optimization problems.

Classical combinatorial optimization problems can be generalized in a natural
way, by considering a related problem relative to a given partition   of the vertices
of the graph. Those \emph{generalized combinatorial optimization problems} have the following primary features:  the graph is given together with a partition of its vertices in clusters and, when 
considering the feasibility constraints of the graph  problem, these
are expressed in relation to the clusters, rather than as individual vertices. Some interesting and intensively studied problems belonging to this category are: the generalized traveling sales person problem \cite{Fisch95,Fisch97,Fisch02}, the generalized vehicle routing problem \cite{GhianiI2000,PopMSP2018}, the partition graph coloring problem, \cite{DemangeMPR2014,DemangeERT2015}, among others. For further references on the category
of generalized combinatorial optimization problems,  we point to
\cite{DrorHC2000,FeremansLL2003,Pop12} and references therein.

Further problems involving graphs and a partition are the \emph{multicolored clique} and the \emph{multicolored independent set} \cite{pietrzak2003,FellowsTCS2009}. In this formalism a partition is seen as a coloring (not necessarily  proper). The goal  of the problems is to select a \emph{multicolored set},  having a vertex of each color,  that induces a clique or an independent set respectively.   The \emph{multicolored clique} problem has been studied  from the parameterized complexity point of view \cite{FellowsTCS2009}. The problem is known to be W[1]-hard, when parameterized by the number of colors, i.e., the number of sets in  the partition.

In this paper we introduce another generalized combinatorial problem, the \emph{multicolored graph realization problem} (\MGRP):  given a graph together with a partition of its vertices (a colored graph), decide  whether there is a multicolored set inducing the cluster graph, i.e., the quotient graph with respect to the given partition.  For example, in data analysis applications, the problem is equivalent to  asking whether we can obtain particular data fulfilling all the inferred relations. 
The \MGRP problem is solvable   in $O(n^k poly(n))$ time, which is polynomial  when the number of colors $k$ is a constant.  But, it is W[1]-hard parameterized by the number of colors, as it includes the multicolored clique problem.   
Observe that, under this parameterization, the cluster graph has constant size and therefore all graph parameters on the cluster graph are constant.

We are interested in analyzing the complexity of the \MGRP problem  when both the number of colors and the treewidth of the cluster graph is unbounded. Our first result, based on the complexity of the Multicolored Clique Problem, is stated as follows (see Section \ref{sec:2} for the appropriate definitions and terminology).
	
\begin{theorem} The MGR problem is W[1]-hard when parameterized by the number of colors or parameterized by the treewidth of the cluster graph .
\end{theorem}
We next focus on specific classes of graphs for which the complexity of natural problems has been widely studied. The first one is   the class of  chordal graphs which form an intensively studied graph class both within structural graph theory
and within algorithmic graph theory.  Recall that  several problems that are hard on other classes of graphs such as graph coloring may be solved in polynomial time on chordal graphs \cite{Gavril1972}.
We show that the MGR problem is
NP-complete for colored graphs whose cluster graph is chordal (see section \ref{sec:3}). The hardness results also hold in case that the number of vertices in a color class is constant.  
\begin{theorem}   The MGR problem is NP-complete, for colored graphs having a
	chordal cluster graph, even when the cluster size is at most 3.
\end{theorem}

Our second family are the chordal bipartite graphs. In particular, we consider the subclasses of convex bipartite graphs and biconvex bipartite graphs which have been used as a benchmark for complexity of homomorphism  problems, see e.g. \cite{enright2014,HuangJP2015,DiazDSS21}.   We show that the MGR problem is
NP-complete for colored graphs whose cluster graph is biconvex bipartite (see section \ref{sec:3}).

\begin{theorem}   The MGR problem is NP-complete, for colored graph having a
	biconvex bipartite cluster graph, even when the cluster size is at most 3.
\end{theorem}

We complement this result showing  that the   {\MGRP} problem belongs to FPT, for  colored graphs having a convex bipartite cluster graph, when parameterized by the size of the clusters and the maximum degree of the non ordered part. 

A third family of bipartite graphs we analyze are $2$--dimensional grid graphs, showing the hardness result in this case as well (see section \ref{sec:4}).

\begin{theorem} The MGR problem is NP-complete, for colored graphs whose
cluster graph is a 2-dimensional grid, even when when the cluster size is 6 and
the input graph has bounded degree. 
\end{theorem}

In view of those results, we analyze the computational complexity of the problem with respect to the size of the color classes. We provide  a complexity dichotomy with respect to this parameter (see section \ref{sec:5}).

\begin{theorem} The MGR problem is NP-complete, for colored graphs with cluster
	size $s\ge 3$, and polynomial time solvable otherwise.
	\end{theorem}

We also show that the \MGRP problem, under the double parameterization cluster size and treewidth of the cluster graph, belongs to FPT.

\setcounter{theorem}{0}
\section{Definitions and preliminaries}\label{sec:2}

 In this section, we provide the  definitions and terminology used in the paper.  We follow notation and basic terminology in graph theory from  Diestel \cite{diestel}.

%
%
\begin{figure}[t]
\begin{center}
\begin{tikzpicture}[scale=0.8,dot/.style={draw,circle,minimum size=1mm,inner sep=0pt,outer sep=0pt}]

\draw (2.5,2.5) node {$y_1$};
\draw (5.5,2.5) node {$y_2$};
\draw (7.5,2.5) node {$y_3$};
\draw (9,2.5) node {$y_4$};

\draw (0,2.5) node {$Y$};
\draw (0,0) node {$X$};

\coordinate [dot,fill =black] (y1) at (2.5,2);
\coordinate [dot,fill =black] (y2) at (5.5,2);
\coordinate [dot,fill =black](y3) at (7.5,2);
\coordinate [dot,fill =black] (y4) at (9,2);

\coordinate [dot,fill =black] (x1) at (1,0.5);
\coordinate [dot,fill =black] (x2) at (2,0.5);
\coordinate [dot,fill =black](x3) at (3,0.5);
\coordinate [dot,fill =black] (x4) at (4,0.5);
\coordinate [dot,fill =black] (x5) at (5,0.5);
\coordinate [dot,fill =black] (x6) at (6,0.5);
\coordinate [dot,fill =black](x7) at (7,0.5);
\coordinate [dot,fill =black](x8) at (8,0.5);
\coordinate [dot,fill =black](x9) at (9,0.5);

\draw (1,0) node {$x_1$};
\draw (2,0) node {$x_2$};
\draw (3,0) node {$x_3$};
\draw (4,0) node {$x_4$};
\draw (5,0) node {$x_5$};
\draw (6,0) node {$x_6$};
\draw (7,0) node {$x_7$};
\draw (8,0) node {$x_8$};
\draw (9,0) node {$x_9$};

\path (y1) edge (x1);
\path (y1) edge (x2);
\path (y1) edge (x3);
\path (y1) edge (x4);
\path (y1) edge (x5);

\path (y2) edge (x4);
\path (y2) edge (x5);
\path (y2) edge (x6);
\path (y2) edge (x7);

\path (y3) edge (x7);
\path (y3) edge (x8);

\path (y4) edge (x5);
\path (y4) edge (x6);
\path (y4) edge (x7);
\path (y4) edge (x8);
\path (y4) edge (x9);

\end{tikzpicture}
\end{center}
\caption{A convex bipartite graph.\label{fig:convex-bgraph}}
\end{figure}
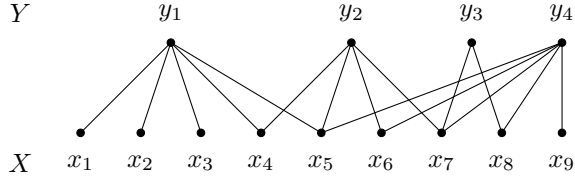

We consider finite, simple and  undirected graphs $G=(V,E)$, i.e.,  without multiple edges or loops.  For $S\subseteq V$, $G[S]$ represents the graph induced by $S$, defined as $G[S]=(S, E\cap\binom{S}{2})$.

A \emph{chordal graph} is one in which all cycles of four or more vertices have a chord. 
A bipartite graph is represented by $G = ( X\cup Y, E)$, where $X$, $Y$ form a bi-partition of the vertex set and $E\subseteq X \times Y$.  Given a  bipartite graph $G = ( X\cup Y, E)$, an ordering of the vertices $X$ has the \emph{adjacency property} (or the ordering is said to be {\it convex}) if, for each vertex $v \in Y$,  $N(v)$ consists of vertices which are consecutive in the ordering of $X$.
{\it Convex bipartite graphs} are the bipartite graphs $G = ( X\cup Y, E)$ that have the adjacency property on one of the partite sets (let us say $X$).   {\it Biconvex bipartite graphs} are the bipartite graphs $G = ( X\cup Y, E)$ that have the adjacency property on both  partite sets. Figure~\ref{fig:convex-bgraph} shows a convex bipartite graph that is not biconvex.   
It is known that there are linear time recognition algorithms for these graphs class (see \cite{nussbaum2010}).

\remove{
A \emph{tree} is a connected acyclic graph, i.e., one not containing any cycles. The vertices of degree 1 in a tree are its \emph{leaves}.   A tree $T$ with a fixed root
$r$ is a rooted tree.   A \emph{clique} is a complete graph.}

\paragraph{Treewidth}   In our results involving treewidth we will use the particular kind of {\em nice} tree--decompositions (for definitions and notation concerning tree decompositions we refer the reader to \cite{CyganMbook2015}). 

Let $(T,X)$ be a tree decomposition of a graph $G$.
We can  make any tree  $T$ into a rooted tree by choosing a node $r\in V (T)$ as the root, and directing all edges to the root.  In this way we can convert a tree decomposition $(T,X)$ into a \emph{rooted tree decomposition}, by fixing one node  $r$ as the root in $T$.
A rooted tree  decomposition $(T,X,r)$ of $G$  allow us to associate to every node in the graph a subgraph of $G$ as follows: For $v \in  V (T)$, let $R_T (v)$ denote the set of nodes in the subtree rooted at $v$ (including $v$).  For $v \in V (T)$, define $V (v)= \cup_{w\in R_T (v)}X_w$ as the set of vertices included in any bag in the subtree rooted at $v$. Finally, define the associated graph as $G(v)= G[V(v)]$, the subgraph induced by $V(v)$. Observe that $G(r)=G$ 
and  that $X_v$ is a separator in $G$.

 A \emph{nice
tree decomposition} is a variant in which the structure of the nodes is simpler. 
A rooted tree decomposition $(T,X)$ is \emph{nice} if each node in $u\in V(T)$ can be classified in one of the following four types.
\begin{itemize}
\item  \emph{start} node: $u$ has no child and $|X_u| = 1$. 

\item \emph{forget} node:  $u$ has one child $v$ and $X_u \subseteq X_v$ and $|X_u| = |X_v | - 1$.
\item   \emph{introduce} node: $u$ has one child $v$ and $X_v \subseteq X_u$ and $|X_u| = |X_v | + 1$. 
\item  \emph{join} node: $u$ has two children $v$ and $w$ with $X_u = X_v = X_w$ .
\end{itemize}
 Given a   tree decomposition of width $k$ for a graph $G$, a rooted nice tree decomposition with width at most $k$ for $G$  and a polynomial number of nodes can be obtained in  $O(kn)$ time  (see for example \cite{CyganMbook2015}).

\paragraph{Complexity classes}

Many NP-complete problems can be
associated with one or more parameterizations. A \emph{parameterization} is a function $\kappa$ assigning a non negative integer value to each input $x$ \cite{FlumG2006}.
A \emph{fixed parameter tractable (FPT)} algorithm  is an algorithm solving a problem  parameterized by $\kappa$ that on input $x$ takes time  $$f(\kappa(x))\cdot |x|^{\Theta(1)}$$ where $f(k)$ is a (super-polynomial) function that does not depend on $n$.
The Parameterized Complexity settles the question of whether a parameterized  problem is solvable by  an FPT algorithm.
 If such an algorithm exists, we say that the parameterized problem belongs to the 
class FPT of fixed parameter tractable problems. In  a series of fundamental papers (see~\cite{DowneyF95-I,DowneyF95-II}), Downey and Fellows introduced a series of complexity classes, namely the classes
$FPT \subseteq {W}[1]\subseteq {W}[2]\subseteq\cdots\subseteq{W}[SAT]\subseteq {W}[P]$,
and proposed special types of reductions such that hardness for 
some of the above classes makes it rather impossible that a problem belongs in FPT.

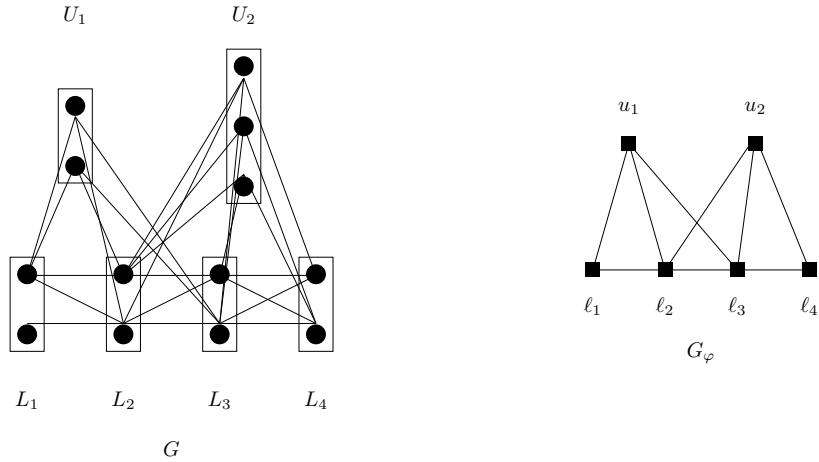
\begin{figure}[t]
\begin{center}\scalebox{0.8}{
\begin{tikzpicture}[scale=0.8]
\draw (2,6) node{$U_1$};
\draw (2,3.5) node[draw] {
\begin{tikzpicture}
\draw(0,0) node[fill=black,circle] {};
\draw(0,1) node[fill=black,circle] {};
\end{tikzpicture}
};

\draw (5.5,6) node{$U_2$};
\draw (5.5,3.7) node[draw] {\begin{tikzpicture}
\draw(0,0) node[fill=black,circle] {};
\draw(0,1) node[fill=black,circle] {};
\draw(0,2) node[fill=black,circle] {};
\end{tikzpicture}
};

\draw (1,-2) node{$L_1$};
\draw (1,0) node[draw] {\begin{tikzpicture}
\draw(0,0) node[fill=black,circle] {};
\draw(0,1) node[fill=black,circle] {};
\end{tikzpicture}};

\draw (3,-2) node{$L_2$};
\draw (3,0) node[draw] {\begin{tikzpicture}
\draw(0,0) node[fill=black,circle] {};
\draw(0,1) node[fill=black,circle] {};
\end{tikzpicture}};

\draw (5,-2) node{$L_3$};
\draw (5,0) node[draw] {\begin{tikzpicture}
\draw(0,0) node[fill=black,circle] {};
\draw(0,1) node[fill=black,circle] {};
\end{tikzpicture}};

\draw (7,-2) node{$L_4$};
\draw (7,0) node[draw] {\begin{tikzpicture}
\draw(0,0) node[fill=black,circle] {};
\draw(0,1) node[fill=black,circle] {};
\end{tikzpicture}};

\remove{\draw (9,-2) node{$L_5$};
\draw (9,0) node[draw] {\begin{tikzpicture}
\draw(0,0.5) node[fill=black,circle]{};
\end{tikzpicture}};}

\remove{
\draw (-1,-2) node{$L_0$};
\draw (-1,0) node[draw] {\begin{tikzpicture}
\draw(-1,0.5) node[fill=black,circle]{};
\end{tikzpicture}};}

\draw (4,-3) node{$G$};


\draw(2,3.9) --  (1,0.6);
\draw(2,2.9) --  (1,0.6);
\draw(2,3.9) --  (3,-0.4);
\draw(2,2.9) --  (3,0.6);
\draw(2,2.9) --  (5,-0.4);
\draw(2,3.9) --  (5,-0.4);

\draw(5.5,3.7) --  (3,0.6);
\draw(5.5,2.7) --  (3,0.6);
\draw(5.5,3.7) --  (5,-0.4);
\draw(5.5,2.7) --  (5,0.6);
\draw(5.5,2.7) --  (7,-0.4);
\draw(5.5,3.7) --  (7,-0.4);
\draw(5.5,4.7) --  (3,0.6);
\draw(5.5,4.7) --  (3,-0.4);
\draw(5.5,4.7) --  (5,-0.4);
\draw(5.5,4.7) --  (7,0.6);

\draw(1,0.6) --  (3,0.6);
\draw(1,-0.4) -- (3,-0.4);
\draw (1,0.6) -- (3,-0.4);

\draw(3,-0.4) --  (5,0.6);
\draw(3,-0.4) -- (5,-0.4);
\draw (3,0.6) -- (5,0.6);

\draw(5,-0.4) --  (7,0.6);
\draw(5,-0.4) -- (7,-0.4);
\draw(5,0.6) -- (7,0.6);
\draw(5,0.6) -- (7,-0.4);


\draw(15,-1) node {$G_\varphi$};
\draw(15,2) node {
\begin{tikzpicture}[scale=0.6]

\draw (2,4.5) node {$u_1$};
\draw (5.5,4.5) node{$u_{2}$};

\draw (2,3.5) node[fill=black] {};
\draw (5.5,3.5) node[fill=black] {};

\draw (1,-1) node {$\ell_1$};
\draw (3,-1) node {$\ell_2$};
\draw (5,-1) node {$\ell_3$};
\draw (7,-1) node {$\ell_4$};

\draw (1,0) node[fill=black] {};
\draw (3,0) node[fill=black] {};
\draw (5,0) node[fill=black] {};
\draw (7,0) node[fill=black] {};

\draw(1,0) --  (3,0);
\draw(3,0) --  (5,0);
\draw(5,0) --  (7,0);

\draw(2,3.5) --  (1,0);
\draw(2,3.5) --  (3,0);
\draw(2,3.5) --  (5,0);

\draw(5.5,3.7) --  (3,0);
\draw(5.5,3.7) --  (5,0);
\draw(5.5,3.7) --  (7,0);

\end{tikzpicture}};
\end{tikzpicture}}
\end{center}
\caption{A colored graph and its  cluster graph.  \label{fig:convex-laygraph}}
\end{figure}

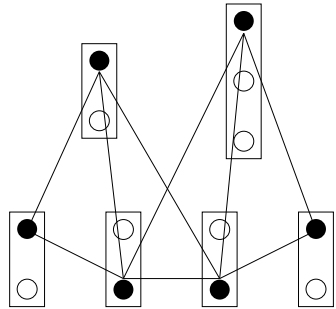
\begin{figure}[t]
\begin{center}
\scalebox{0.8}{
\begin{tikzpicture}[scale=0.8]

\draw (2.5,3.5) node[draw] {
\begin{tikzpicture}
\draw(0,0) node[draw,circle] {};
\draw(0,1) node[fill=black,circle] {};
\end{tikzpicture}
};
\draw (5.5,3.7) node[draw] {\begin{tikzpicture}
\draw(0,0) node[draw,circle] {};
\draw(0,1) node[draw,circle] {};
\draw(0,2) node[fill=black,circle] {};
\end{tikzpicture}
};

\draw (1,0) node[draw] {\begin{tikzpicture}
\draw(0,0) node[draw,circle] {};
\draw(0,1) node[fill=black,circle] {};
\end{tikzpicture}};
\draw (3,0) node[draw] {\begin{tikzpicture}
\draw(0,0) node[fill=black,circle] {};
\draw(0,1) node[draw,circle] {};
\end{tikzpicture}};
\draw (5,0) node[draw] {\begin{tikzpicture}
\draw(0,0) node[fill=black,circle] {};
\draw(0,1) node[draw,circle] {};
\end{tikzpicture}};
\draw (7,0) node[draw] {\begin{tikzpicture}
\draw(0,0) node[draw,circle] {};
\draw(0,1) node[fill=black,circle] {};
\end{tikzpicture}};

\remove{
\draw (9,0) node[draw] {\begin{tikzpicture}
\draw(0,0.5) node[fill=black,circle]{};
\end{tikzpicture}};

\draw (-1,0) node[draw] {\begin{tikzpicture}
\draw(-1,0.5) node[fill=black,circle]{};
\end{tikzpicture}};
}

\draw(2.5,3.9) --  (1,0.6);
\draw(2.5,3.9) --  (3,-0.4);
\draw(2.5,3.9) --  (5,-0.4);

\draw(5.5,4.7) --  (3,-0.4);
\draw(5.5,4.7) --  (5,-0.4);
\draw(5.5,4.7) --  (7,0.6);


\draw (1,0.6) -- (3,-0.4);

\draw(3,-0.4) -- (5,-0.4);

\draw(5,-0.4) --  (7,0.6);


\end{tikzpicture}}
\end{center}
\caption{In black a   multicolored set $S$ realizing $G_\varphi$, for the colored convex layered graph given in Figure \ref{fig:convex-laygraph}. \label{fig:convex-laygraph-ske}}
\end{figure}

\paragraph{The multicolored graph realization problem}
A \emph{coloring} of  a graph $G=(V, E)$ is a map $\varphi:V\to \mathbb{N}$. Observe that our colorings might not be proper as we do not require that adjacent vertices get different colors. However, all our results will also hold for proper colorings. Given a coloring $\varphi$ on $G$, let $k(G,\vp)$ be the number of different colors used by $\vp$. 
\remove{ that uses exactly $k$ colors, we use the term \emph{$k$-coloring} and w.l.o.g. assume that $\varphi:V\to [k]$. So, we assume that in a $k$-coloring of $G$ each one of the colors in $[k]$ is used to color at least one vertex.}  
We use the term \emph{colored graph} to refer to a pair $(G,\vp)$.

 Given a colored  graph $(G,\vp)$  with $G=(V,E)$,   we say that two vertices $u$ and $v$ are equivalent whenever they get  the same color, i.e., $u\sim_\varphi v$ iff $\varphi(u)=\varphi(v)$. This is a natural equivalence relation that partitions the vertices of $G$ into non empty color classes. We use $[u]_\vp$ (or just $[u]$ when $\vp$ is clear from the context) to denote the color class of $u$.

Given a colored  graph $(G,\vp)$, with color classes $A_1, \dots , A_k$, the associated \emph{cluster graph} is the graph $G_\vp=(\{1,\dots, k\},E_\vp)$ where, for $i,j\in \{1,\dots, k\}$ with $i\neq j$, there is an edge $(i,j)\in E_\vp$ whenever there is and edge $(u,v)\in E$ with $u\in A_i$ and $v \in A_j$. Figure \ref{fig:convex-laygraph} gives an example of a colored graph, each rectangle representing a color class,  and its associated cluster graph.  We say that a set of vertices $S\subseteq V$ is \emph{multicolored} if, for any $1\leq i\leq  k$, we have that  $|S\cap A_i|= 1$, i.e., there is exactly one vertex in $S$ from each color class.   For a multicolored set $S$, we assume that $S=\{u_1,\dots, u_k\}$ so that, for $1\leq i\leq k$, $[u_i]= A_i$.
A \emph{multicolored realization} of $G_\vp$ is a multicolored  subset $S\subseteq V$ such that the restriction of $\varphi$ to $S$ is an isomorphism between $G[S]$ and $G_\vp$, i.e.,  for $u_i,u_j\in S$, $(u_i,u_j)\in E(G)$ if and only if $(i,j)\in E(G_\vp)$. Figure \ref{fig:convex-laygraph-ske} shows a multicolored set realizing $G_\varphi$, for the colored convex layered graph given in Figure \ref{fig:convex-laygraph}.

With this notation we can state formally the definition of our problem. 

\smallskip
\begin{quote}
\emph{Multicolored graph realization} problem (\MGRP)\\
Given a graph $G=(V,E)$ together with a coloring $\varphi$, does there exists a  multicolored realization of $G_\vp$?  
\end{quote}

Observe that, given a colored graph $(G, \varphi)$ and a multicolored set $S$, we can check, in polynomial time,  whether $S$ is a realization of $G_\varphi$. Therefore, the  \MGRP  problem belongs to NP.

We are interested in analyzing the computational complexity of the \MGRP problem  under different parameterizations:  
the \emph{number of used colors} $k(G,\varphi)$, the   \emph{cluster size}
$s(G,\vp) = \max_{v\in V} |\vp^{-1} (v)|$, the treewidth of the cluster graph,  and other combinations of  parameters.  Although  the problem is defined as usual in its decision form,  the  algorithms provided in the paper are constructive, they produce a multicolored realization in the case that one exists. 

Our first results follows from a well known problem the Multicolored clique Problem (also known as  the Partitioned clique problem)  which according to \cite{CyganMbook2015} was introduced in \cite{FellowsTCS2009,pietrzak2003}. We provide here a formal definition of the problem adapted to our notation.
\begin{quote}
\emph{Multicolored clique} problem  (MC)\\
Given a  colored graph $(G,\vp)$, does there exists a a multicolored set $S$ such that $G[S]$ is a clique? 
\end{quote}
The MC problem is known to be $W[1]$-hard parameterized by the number of colors \cite{FellowsTCS2009}. This yields the following result in the multicolored problem language.

\begin{theorem}\label{thm:btw}
 The {\MGRP} problem  is $W[1]$-hard  when parameterized by the number of colors or parameterized by the treewidth of  the cluster graph.  
\end{theorem}
\begin{proof}
Observe that, in a colored graph, a necessary condition to have a multicolored clique is the cluster graph $G_\vp$ being itself a clique.  Therefore, the MC problem is the particular case of the \MGRP problem when the cluster graph $G_\vp$ is a complete graph.  
On the other hand, when the number of used colors $k$ is bounded, the cluster graph  $G_\varphi$ has constant size and  therefore has bounded treewidth.  
\end{proof}

\section{Chordal and convex bipartite cluster graphs}\label{sec:skeleton}\label{sec:3}

In this section we analyze the complexity of the \MGRP problem on colored graphs having a  cluster graph that is chordal or chordal bipartite. We start presenting the NP-hardness for the case of chordal graphs. 
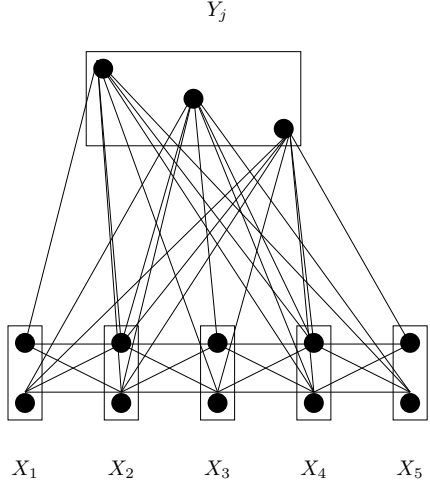
\begin{figure}[t]
\begin{center}
\scalebox{0.8}{
\begin{tikzpicture}[scale=0.8]

\draw(5,7.5) node {$Y_j$};
\draw (4.5,5.7) node[draw] {\begin{tikzpicture}
\draw(3,0) node[fill=black,circle] {};
\draw(1.5,0.5) node[fill=black,circle] {};
\draw(0,1) node[fill=black,circle] {};
\end{tikzpicture}
};

\draw (1,0) node[draw] {\begin{tikzpicture}
\draw(0,0) node[fill=black,circle] {};
\draw(0,1) node[fill=black,circle] {};
\end{tikzpicture}};
\draw (3,0) node[draw] {\begin{tikzpicture}
\draw(0,0) node[fill=black,circle] {};
\draw(0,1) node[fill=black,circle] {};
\end{tikzpicture}};
\draw (5,0) node[draw] {\begin{tikzpicture}
\draw(0,0) node[fill=black,circle] {};
\draw(0,1) node[fill=black,circle] {};
\end{tikzpicture}};
\draw (7,0) node[draw] {\begin{tikzpicture}
\draw(0,0) node[fill=black,circle] {};
\draw(0,1) node[fill=black,circle] {};
\end{tikzpicture}};
\draw (9,0) node[draw] {\begin{tikzpicture}
\draw(0,0) node[fill=black,circle] {};
\draw(0,1) node[fill=black,circle] {};
\end{tikzpicture}};


\remove{\draw (15,0) node[draw] {\begin{tikzpicture}
\draw(0,0.5) node[fill=black,circle]{};
\end{tikzpicture}};

\draw (-1,0) node[draw] {\begin{tikzpicture}
\draw(-1,0.5) node[fill=black,circle]{};
\end{tikzpicture}};
}


\draw(1,-2) node {$X_1$};
\draw(3,-2) node {$X_2$};
\draw(5,-2) node {$X_3$};
\draw(7,-2) node {$X_4$};
\draw(9,-2) node {$X_5$};


\draw(2.5,6.5) --  (1,0.6);
\draw(2.5,6.5) --  (3,0.6);
\draw(2.5,6.5) --  (3,-0.4);
\draw(2.5,6.5) --  (5,-0.4);
\draw(2.5,6.5) --  (7,0.6);
\draw(2.5,6.5) --  (7,-0.4);
\draw(2.5,6.5) --  (9,-0.4);

\draw(4.5,5.8) --  (1,-0.4);
\draw(4.5,5.8)--  (3,0.6);
\draw(4.5,5.8) --  (3,-0.4);
\draw(4.5,5.8) --  (5,0.6);
\draw(4.5,5.8) --  (7,0.6);
\draw(4.5,5.8) --  (7,-0.4);
\draw(4.5,5.8) --  (9,-0.4);

\draw(6.5,5)  --  (1,-0.4);
\draw(6.5,5) --  (3,0.6);
\draw(6.5,5)--  (3,-0.4);
\draw(6.5,5) --  (5,-0.4);
\draw(6.5,5) --  (7,0.6);
\draw(6.5,5) --  (7,-0.4);
\draw(6.5,5) --  (9,0.6);


\draw(1,-0.4) --  (3,0.6);
\draw(1,0.6) --  (3,0.6);
\draw(1,-0.4) -- (3,-0.4);
\draw (1,0.6) -- (3,-0.4);

\draw(3,-0.4) --  (5,0.6);
\draw(3,-0.4) -- (5,-0.4);
\draw (3,0.6) -- (5,0.6);
\draw (3,0.6) -- (5,-0.4);

\draw(5,-0.4) --  (7,0.6);
\draw(5,-0.4) -- (7,-0.4);
\draw(5,0.6) -- (7,0.6);
\draw(5,0.6) -- (7,-0.4);

\draw(7,-0.4) --  (9,0.6);
\draw(7,-0.4) -- (9,-0.4);
\draw(7,0.6) -- (9,0.6);
\draw(7,0.6) -- (9,-0.4);


\end{tikzpicture}}
\end{center}
\caption{The connections representing clause  $C_j=(x_1, x_3,x_5)$  in the associated  colored graph.  \label{fig:convex-graph-clause}}
\end{figure}

\begin{theorem}\label{thm:ConSke-NPhard}
 The {\MGRP} problem is NP-complete, for colored  graphs having a chordal cluster graph, even when the cluster size is at most 3. 
\end{theorem}
\begin{proof}
As the \MGRP problem belongs to NP,  we have only to show that the problem is NP-hard. For doing so, we provide a reduction from  1-in-3 Monotone SAT problem which is known to be NP-complete \cite{schaefer1978}.  The input formula $\Phi$ (in CNF) is the conjunction of  $m$ clauses,  $C_1, C_2, \ldots, C_m$,  over a set of $n$ variables. Furthermore,  each  clause is the disjunction of exactly three non-negated variables. The problem asks whether  it is possible to assign to each  of the $n$ variables $x_1, x_2 \ldots x_n$ a value in $\{0,1\}$,  so that, in each clause, exactly one of the three variables is set to 1. As a shorthand, we write each clause as $(x,y,z)$, without use of the or symbol.   Furthermore, we assume that the variables in a clause are given in increasing order, i.e. $C_j = (x_{j_1},x_{j_2}, x_{j_3})$ with  $j_1\leq j_2\leq j_3$

Given an input  formula $\Phi$ to the 1-in-3 Monotone SAT problem, we construct a colored convex layered graph $(G,\varphi)$. Instead of describing the coloring $\varphi$, we provide the different color classes and the connection among their vertices. In this way, it is easy to see that the construction provides a colored graph with a chordal cluster graph. For each variable $x_i$, $1\leq i\leq n$, we create a color class  $X_i$ having two vertices $v_i^0$ and $v_i^1$.  For each clause $C_j$, $1\leq j\leq m$, we create a color class $Y_j=\{ u_j^{100}, u_j^{010},u_j^{001}\}$ with one vertex for each possible join assignment with only one 1.  

The edge set in $G$ is the following:
\begin{enumerate}
\item For $1\leq i<n$, we connect the vertices in $X_i$ with those in $X_{i+1}$ by a complete bipartite subgraph. 
\item For a clause $C_j=(x_{j_1},x_{j_2}, x_{j_3})$, we connect: 
\begin{itemize}
\item $u_j^{100}$ with $v_{j_1}^1$, $v_{j_2}^0$ and $v_{j_3}^0$ and with both $v_i^0$ and $v_i^1$, for $j_1<i<j_2$ and $j_2<i<j_3$; 
\item $u_j^{010}$ with $v_{j_1}^0$, $v_{j_2}^1$ and $v_{j_3}^0$ and with both $v_i^0$ and $v_i^1$, for $j_1<i<j_2$ and $j_2<i<j_3$; 
\item $u_j^{001}$ with $v_{j_1}^0$, $v_{j_2}^0$ and $v_{j_3}^1$ and with both $v_i^0$ and $v_i^1$, for $j_1<i<j_2$ and $j_2<i<j_3$.
\end{itemize} 
\end{enumerate}
Figure  \ref{fig:convex-graph-clause} shows  the color classes, the vertices,  and the connections corresponding to a  clause.
It is easy to see that $G_\vp$ is a chordal graph.

Now we show that the construction is indeed a reduction from the 1-in-3 Monotone SAT problem to the \MGRP problem. As a main argument, we translate   the appearance of $v_i^0$ in a multicolored realization  (if it exists) as the assignment $x_i=0$ and the appearance of $v_i^1$ as the assignment $x_i=1$, and viceversa.  

Let us assume that we have an  assignment $T$ such that $T(x_i)=t_i\in \{0,1\}$, for $1\leq i\leq n$, in which exactly one variable in each clause in $\Phi$ is assigned  to 1.  Consider the set 
\[S=\{v_i^{t_i}\mid 1\leq i\leq n\} \cup \{u_j^{ t_{j_1}t_{j_2}t_{j_3}}\mid 1\leq j\leq m\}.\]
For a clause $C_j=(x_{j_1},x_{j_2},x_{j_3})$, by the definition of $G$,  the vertex $u_j^{ t_{j_1}t_{j_2}t_{j_3}}$ becomes connected with all the vertices $v_i^{t_i}$, for  $i\in [j_1, j_3]$. Therefore, $S$ is a  multicolored realization of $G_\varphi$.

For the reverse implication,  let $X=\cup_{i=1}^n X_i$.  Assume that $S$ is a multicolored realization of $G_\vp$ and that  $S\cap X =\{v_1^{t_1},v_2^{t_2},\ldots v_n^{t_n}\}$ where $t_i\in \{0,1\}$, for $1\leq i \leq n$. Consider the assignment  $T(x_i)=t_i$. By construction, the assignment $T$ is correct, as each variable $x_i$ gets a unique assigned value. 
Furthermore, observe that, if  a clause $C_j=(x_{j_1},x_{j_2},x_{j_3})$ contains a variable with assigned value 1, the other variables in $C_j$ are assigned value 0. This is due to the fact that $S$ is a  realization of the cluster graph. Assuming w.l.o.g. that $t_{j_1}=1$,   $G[S]$  contains $v_{j_1}^{1}$, then it must contain $v_{j_2}^{0}$ and $v_{j_3}^{0}$ as well, as  $u_j^{100}$ must belong to $S$ and its neighbors must be contiguous on the interval $[j_1, j_3]$. Thus in each clause only one of the three variables has assigned  value 1. 

Finally, observe that the graph can be constructed in polynomial time from the given formula. 
\end{proof}

We can adapt the previous reduction to show that the \MGRP problem remains hard  when the  cluster graph is a convex bipartite or a biconvex bipartite

\begin{theorem}\label{thm:ConBip-NPhard}
The  {\MGRP} problem is NP-complete, for colored graph having a biconvex bipartite cluster graph, even when the cluster size is at most 3. 
\end{theorem}
\begin{proof}

We  modify slightly the construction in the previous theorem to get a colored graph whose cluster graph is a convex bipartite graph. For doing so, we just remove the connections among the lower layers. In this way, the cluster graph is a bipartite graph, that indeed is convex with respect to the $X$ color classes. As the connections among consecutive  $X$ clusters  were all-to-all, any multicolored subset  realizing the cluster graph does it in both graphs.      

Finally, we add some additional connections to the vertices in the clusters corresponding to a clause. For a clause $C_j=(x_{j_1},x_{j_2}, x_{j_3})$, we connect $u_j^{100}$, $u_j^{010}$ and $u_j^{001}$ with both $v_i^0$ and $v_i^1$, for $1\leq i <j_1$ and   $j_3<i\leq n$. As the new connections are the same as those for the intermediate variables not appearing in $C_j$ the construction is    a reduction from the 1-in-3 Monotone SAT problem to the \MGRP problem. Observe, that the graph $G_\vp$ is a complete bipartite graph and therefore it is biconvex. 
\end{proof}

Let $G=(X\cup Y , E)$ be a  convex bipartite graph that has the adjacency property with respect to $X$. We define the \emph{spread}  of $G$ as the maximum degree of the vertices in $Y$.
Although Theorem~\ref{thm:ConSke-NPhard} shows the hardness of  the {\MGRP} problem even when the cluster size is bounded,  the spread is unbounded. Our next results gives an FPT algorithm solving  the \MGRP problem  on colored convex bipartite graphs when parameterized by both, the size of the cluster size and the spread.

\begin{proposition}\label{prop:fpt-cb} The  {\MGRP} problem belongs to FPT, for  colored graphs having a convex bipartite cluster graph, when parameterized by the cluster size and the spread.   
  \end{proposition} 
\begin{proof}
Let $(G,\vp)$ be a colored convex bipartite graph, with cluster size $\ell$ and spread $d$.
Assume that the color classes are $X_1,\dots,X_\alpha$, $Y_1,\dots,Y_\beta$ and that 
$G_\vp$ is a bipartite cluster graph having the adjacency property with respect to the clusters $X_1,\dots,X_\alpha$.   Recall that there is a  linear time algorithm  to recognize convex bipartite graphs \cite{spinrad1987, nussbaum2010}. This  allows us to obtain the ordering on the $X$ part in linear time. We devise a dynamic programming algorithm based on this  ordering.  

Let $X=\cup_{i=1}^\alpha X_i$ and  $Y=\cup_{j=1}^\beta Y_j$. 
For each $i$, $d\leq i\leq \alpha$, let $P_i= X_{i-d} \times \dots \times  X_i$, of tuples formed by $d-1$ vertices in consecutive layers ending at  a vertex in $X_i$.  Let $G_i$ be the subgraph induced in $G$ by 
$$ V_i = \left(X_1\cup \dots\cup X_i\right) \cup \left(\cup_{j\mid b_j\leq i} Y_j\right ).$$ 

For each $1<i \leq \alpha$, our  dynamic programming algorithm keeps a table $M_i$ holding a  boolean value for each $p\in P_i$. Thus the table size is $|P_i|$.  The entry $M_i(p)$ will be set to 1  whenever  there is a  multicolored set $S\subseteq V_i$ that is a realization for $G_{i,\vp}$ such that $S$ contains all the vertices in  $p$. Otherwise, the value will be 0.

When $i=d$, for each $p\in P_d$, we have to check whether the set of vertices in $p$ can be extended to a multicolored realization in $G_d$. For this, it is enough to check whether, for each color class  $Y_j$ included in $V_d$,  there exists  a vertex $u_j\in Y_j$ so that it is connected to all the vertices in $p$ belonging to  layers in $[a_j,b_j]$. In this  case, set $M_d(p)= 1$, and otherwise set  $M_d(p)= 0$. 
  
When $d < i\leq \alpha$, for $p\in P_i$,  we set $M_i(p)=1$,    if 
\begin{itemize}
\item[(1)] for any $Y_j$ with $b_j=i$,   there is a   a vertex  in $P_j$ connected to all the vertices in $p$ in color classes $X_i$ with $i\in [a_j,b_j]$, and
\item[(2)] there exists  $p'\in  P_{i-1}$, such that $p$ extends $p'$  and $M_{i-1}(p) = 1$. 
 \end{itemize}
 
Observe that, as the considered $Y$ sets are  included in $G_i$ but not in $G_{i-1}$, condition (1),  guarantees that $p$ can be extended to a multicolored realization with respect to the newly incorporated $Y$-sets. On the other hand, condition (2) guarantees that the vertices in $p$ can be extended  to a multicolored realization with respect to $G_{i-1}$, as $p'$ contains all the vertices in  $p$ except the one in the last $X$ cluster.   

So, we can conclude that the proposed algorithm correctly computes $M_i(p)$,  for each $i$,  $d \leq i\leq \alpha$, and $p\in P_i$.

Note that  $G_\alpha= G$, so if $M_\alpha(p)=1$, the set of vertices in $p$ can be extended to a multicolored realization for $G$.   The last step of our algorithm just checks  whether there is  $p\in P_d$ having $M_d(p)=1$.

For the time complexity, observe that $|P_i|\leq \ell^d$. Checking conditions (1) and (2) is the most costly operation.  For a given $j$ and $p$, checking condition (1)  takes time $O(dn)$. Furthermore, the algorithm  performs this checking once for each $j$. So the overall time is $O(\ell^d dn)$. For given $p\in P_i$ and $p'\in P_{i-1}$,  checking condition (2) requires constant time. On the other hand the number of tuples that $p$ can extend is at most $|X_{i-1}|$. This gives an overall time of $O(\ell^2 \ell^d)$. 
The total cost is $O(  (n d + \ell^2)    \ell^d)$. 

Using standard dynamic programming techniques the algorithm can be adapted to produce a multicolored realization when one exists within the same time bounds. 
\end{proof}

\section{Grid cluster graphs}\label{sec:4}

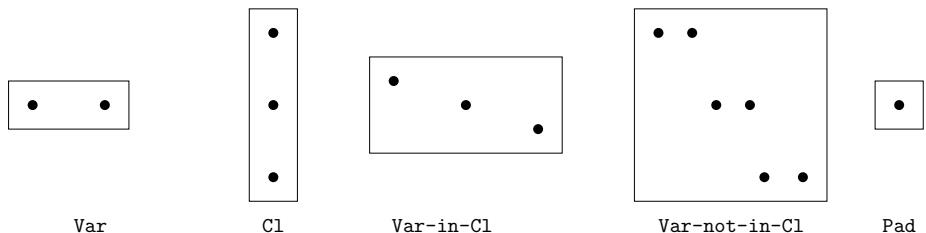
\begin{figure}[t]
\begin{center}
\scalebox{0.8}{
\begin{tikzpicture}[scale=0.8] 

\draw (1.5,1.5) node[fill=black,circle,scale=0.5]  (x0) {};
\draw(0,1.5) node[fill=black,circle,scale=0.5] (x1) {};
\draw (-0.5,1) rectangle +(2.5,1);

\draw(5,0) node[fill=black,circle,scale=0.5] (xi3) {};
\draw(5,1.5) node[fill=black,circle,scale=0.5] (xi2) {};
\draw(5,3) node[fill=black,circle,scale=0.5] (xi1){};
\draw (4.5,-0.5) rectangle +(1,4);

\draw(10.5,1) node[fill=black,circle,scale=0.5] (xi3) {};
\draw(9,1.5) node[fill=black,circle,scale=0.5] (xi2) {};
\draw(7.5,2) node[fill=black,circle,scale=0.5] (xi1){};
\draw (7,0.5) rectangle +(4,2);


\draw(15.2,0) node[fill=black,circle,scale=0.5] (xi31) {};
\draw(14.2,1.5) node[fill=black,circle,scale=0.5] (xi21) {};
\draw(13,3) node[fill=black,circle,scale=0.5] (xi11){};
\draw(16,0) node[fill=black,circle,scale=0.5] (xi30) {};
\draw(14.9,1.5) node[fill=black,circle,scale=0.5] (xi20) {};
\draw(13.7,3) node[fill=black,circle,scale=0.5] (xi10){};
\draw (12.5,-0.5) rectangle +(4,4);

\draw(18,1.5) node[fill=black,circle,scale=0.5] (a) {};
\draw(17.5,1) rectangle +(1,1);

\draw (1.2,-1) node {\gn{Var}};
\draw (5,-1) node {\gn{Cl}};
\draw (8.5,-1) node {\gn{Var-in-Cl}};
\draw (14.5,-1) node {\gn{Var-not-in-Cl}};
\draw (18,-1) node {\gn{Pad}};
\end{tikzpicture}}

\end{center}
\caption{The five basic cluster gadgets.  \label{fig:grid-gadgets}}
\end{figure}
\begin{figure}[t]
\begin{center}
\scalebox{0.7}{
\begin{tikzpicture}[every node/.style={scale=0.5},scale=0.8] 

\draw (2.5,7) node[fill=black,circle]  (x0) {};
\draw(1,7) node[fill=black,circle] (x1) {};

\draw(3,3) node[fill=black,circle] (xi3) {};
\draw(1.5,3.5) node[fill=black,circle] (xi2) {};
\draw(0,4) node[fill=black,circle] (xi1){};

\draw(2.5,0) node[fill=black,circle] (x00) {};
\draw(1,0) node[fill=black,circle] (x11) {};

\draw (0.5,-0.5) rectangle +(2.5,1);
\draw (-0.5,2.5) rectangle +(4,2);
\draw (0.5,6.5) rectangle +(2.5,1);

\draw (x1) -- (xi1) -- (x11);
\draw (x0) -- (xi2) -- (x00);
\draw (x0) -- (xi3) -- (x00);

\end{tikzpicture}
\hspace{0.5cm}
\begin{tikzpicture}[every node/.style={scale=0.5},scale=0.8] 

\draw (2.5,7) node[fill=black,circle]  (x0) {};
\draw(1,7) node[fill=black,circle] (x1) {};

\draw(3,3) node[fill=black,circle] (xi3) {};
\draw(1.5,3.5) node[fill=black,circle] (xi2) {};
\draw(0,4) node[fill=black,circle] (xi1){};

\draw(2.5,0) node[fill=black,circle] (x00) {};
\draw(1,0) node[fill=black,circle] (x11) {};

\draw (0.5,-0.5) rectangle +(2.5,1);
\draw (-0.5,2.5) rectangle +(4,2);
\draw (0.5,6.5) rectangle +(2.5,1);

\draw (x1) -- (xi2) -- (x11);
\draw (x0) -- (xi1) -- (x00);
\draw (x0) -- (xi3) -- (x00);
\end{tikzpicture}
\hspace{0.5cm}
\begin{tikzpicture}[every node/.style={scale=0.5},scale=0.8] 

\draw (2.5,7) node[fill=black,circle]  (x0) {};
\draw(1,7) node[fill=black,circle] (x1) {};

\draw(3,3) node[fill=black,circle] (xi3) {};
\draw(1.5,3.5) node[fill=black,circle] (xi2) {};
\draw(0,4) node[fill=black,circle] (xi1){};

\draw(2.5,0) node[fill=black,circle] (x00) {};
\draw(1,0) node[fill=black,circle] (x11) {};

\draw (0.5,-0.5) rectangle +(2.5,1);
\draw (-0.5,2.5) rectangle +(4,2);
\draw (0.5,6.5) rectangle +(2.5,1);

\draw (x1) -- (xi3) -- (x11);
\draw (x0) -- (xi2) -- (x00);
\draw (x0) -- (xi1) -- (x00);

\end{tikzpicture}
\hspace{0.5cm}
\begin{tikzpicture}[every node/.style={scale=0.5},scale=0.8] 

\draw (2.5,7) node[fill=black,circle]  (x0) {};
\draw(1,7) node[fill=black,circle] (x1) {};

\draw(3.2,3) node[fill=black,circle] (xi31) {};
\draw(1.8,3.7) node[fill=black,circle] (xi21) {};
\draw(0.2,4.5) node[fill=black,circle] (xi11){};
\draw(2.5,3) node[fill=black,circle] (xi30) {};
\draw(1.2,3.7) node[fill=black,circle] (xi20) {};
\draw(-0.5,4.5) node[fill=black,circle] (xi10){};

\draw(2.5,0) node[fill=black,circle] (x00) {};
\draw(1,0) node[fill=black,circle] (x11) {};

\draw (0.5,-0.5) rectangle +(2.5,1);
\draw (-1,2) rectangle +(5,3);
\draw (0.5,6.5) rectangle +(2.5,1);

\draw (x0) -- (xi31) -- (x00);
\draw (x0) -- (xi21) -- (x00);
\draw (x0) -- (xi11) -- (x00);

\draw (x1) -- (xi30) -- (x11);
\draw (x1) -- (xi20) -- (x11);
\draw (x1) -- (xi10) -- (x11);

\end{tikzpicture} }

\vskip 0.5cm
\begin{tikzpicture}[scale= 0.8]
\node (0,0) {\phantom{a}};
\draw (1.5,0) node  {$x_{j_1}$};
\draw (4.5,0) node  {$x_{j_2}$};
\draw (8,0) node  {$x_{j_3}$};
\draw (12.5,0) node  {$x_{i}\notin C$};
\node (25,0) {\,};
\end{tikzpicture} 
\end{center}
\caption{The vertical connections among contiguous  \gn{Var}--\gn{Var-in-Cl} and \gn{Var}--\gn{Var-not-in-Cl} gadgets,  for a clause  $C_j=(x_{j_1}, x_{j_2},x_{j_3})$.  \label{fig:grid-gadget-Vert}}
\end{figure}
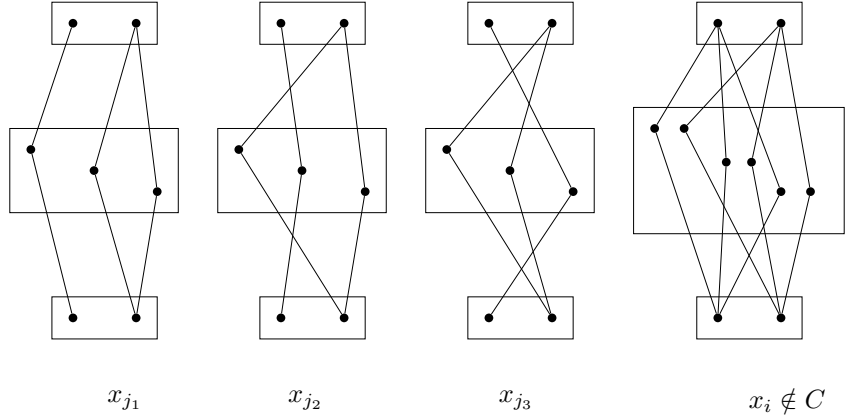

In this section we consider  the \MGRP problem restricted to colored graphs $(G,\vp)$ for which the resulting cluster graph $G_\vp$ is a 2-dimensional grid. Recall that  a \emph{two-dimensional grid graph},  is a lattice graph, obtained as the  Cartesian product of two path graphs, on $n$ and $m$ vertices respectively. Formally,  an $n\times m $ grid graph $L_{n,m}$ has vertex set $\{(i,j)\mid 1\leq i\leq n, 1\leq j\leq m \}$. Two vertices $(i,j)$ and $(i',j')$  are adjacent if and only if $|i - i' | + |j - j'| = 1$. To show that the problem is hard in this case,  we again provide a reduction from  the 1-in-3 Monotone 3-SAT problem. 

\begin{theorem}\label{thm:Grid-NPhard}
The  {\MGRP} problem is NP-complete, for colored graphs whose  cluster graph is a 2-dimensional grid, even when when the cluster size is 6 and the input graph  has bounded degree.  
\end{theorem}
\begin{proof}
Let us describe a reduction from the 1-in-3 Monotone 3-SAT problem.  Assume that $\Phi$ is a monotone 3-SAT formula on   $n$ variables $x_1,\dots,x_n$ having $m$ clauses  $C_1,\dots,C_m$ each with exactly three variables. For simplicity, as we did before,  we assume that  clause $j$, $1\leq j\leq m$, has the form $C_j=(x_{j_1},x_{j_2},x_{j_3})$, with      
 $j_1< j_2< j_3$. 
 We construct a colored graph $(G,\varphi)$  in polynomial time from $\Phi$ and will show that $\Phi$ has a valid truth assignment if and only if  $G_\vp$ has a multicolored realization. 

Our construction uses several gadgets, each one of them describing a color class of $(G,\vp)$, see Figure~\ref{fig:grid-gadgets}. The graph $G$ will be formed by several copies of those gadgets. We locate them inside a 2-dimensional grid, as shown in Figure~\ref{fig:grid-example}.  In this way, it will be clear that $G_\vp$ is indeed a $(n+2)\times (2m+1)$ 2-dimensional grid 

\paragraph{The gadgets}
We use 5 basic gadgets, each used gadget  constitutes  a color class  in $G$ (see Figure~\ref{fig:grid-gadgets}). 
The first kind, the \gn{Var} gadget,  contains $2$ vertices, we refer to them for their positions in the box, left and right. They are used to represent a variable and a  selection of one vertex will correspond to an assignment of value to the variable, left vertex with a  1 and the right one with a 0.  The second gadget,  the \gn{Cl} gadget,  contains $3$ vertices, we refer to them as the upper, middle and lower vertices. They are used to represent a clause. Those three nodes are used to be identify the three valid assignment values for the variables in the clause, upper with 100, middle with 010, and lower with 001. The third and the four  gadgets,  the \gn{Var-in-Cl}  and the \gn{Var-not-in-Cl} gadgets,  contain 3 and 6 vertices respectively. The first one is used for a variable that appears in a clause and the second when it does not appear.  The three node in the \gn{Var-in-Cl} will be referred as upper, middle and lower.  The \gn{Var-not-in-Cl}  block has three groups of two vertices (upper, middle and lower groups), inside each group, we use position (left or right) as reference. Finally, the fifth gadget, the \gn{Pad} gadget contains only one vertex.

\begin{figure}[t]
\begin{center}
\begin{subfigure}[b]{\textwidth}
\centering
\scalebox{0.8}{\vbox{
\begin{tikzpicture}[every node/.style={scale=0.7},scale=0.6] 

\draw(1,0) node[fill=black,circle] (xi3) {};
\draw(1,1.5) node[fill=black,circle] (xi2) {};
\draw(1,3) node[fill=black,circle] (xi1){};
\draw (0.5,-0.5) rectangle +(1,4);

\draw(7,1) node[fill=black,circle] (xi31) {};
\draw(5.5,1.5) node[fill=black,circle] (xi21) {};
\draw(4,2) node[fill=black,circle] (xi11){};
\draw (3.5,0.5) rectangle +(4,2);

\draw (xi1) -- (xi11);
\draw (xi2) -- (xi21);
\draw (xi3) -- (xi31);
\end{tikzpicture}
\hspace{1cm}
\begin{tikzpicture}[every node/.style={scale=0.7},scale=0.6] 

\draw(11,0) node[fill=black,circle] (xi3) {};
\draw(11,1.5) node[fill=black,circle] (xi2) {};
\draw(11,3) node[fill=black,circle] (xi1){};
\draw (10.5,-0.5) rectangle +(1,4);

\draw(7,1) node[fill=black,circle] (xi31) {};
\draw(5.5,1.5) node[fill=black,circle] (xi21) {};
\draw(4,2) node[fill=black,circle] (xi11){};
\draw (3.5,0.5) rectangle +(4,2);

\draw (xi1) -- (xi11);
\draw (xi2) -- (xi21);
\draw (xi3) -- (xi31);
\end{tikzpicture}}}
\vspace{0.25cm}
\caption{\gn{Cl} -- \gn{Var-in-Cl} and \gn{Var-in-Cl} -- \gn{Cl}}
\end{subfigure}

\begin{subfigure}[b]{\textwidth}
\centering
\scalebox{0.8}{\vbox{
\vspace{0.25cm}
\begin{tikzpicture}[every node/.style={scale=0.7},scale=0.6] 

\draw(1,0) node[fill=black,circle] (xi3) {};
\draw(1,1.5) node[fill=black,circle] (xi2) {};
\draw(1,3) node[fill=black,circle] (xi1){};
\draw (0.5,-0.5) rectangle +(1,4);

\draw(7.5,0) node[fill=black,circle] (xi31) {};
\draw(6,1.5) node[fill=black,circle] (xi21) {};
\draw(4.5,3) node[fill=black,circle] (xi11){};
\draw(6.5,-0.5) node[fill=black,circle] (xi30) {};
\draw(5,1) node[fill=black,circle] (xi20) {};
\draw(3.5,2.5) node[fill=black,circle] (xi10){};
\draw (3,-1) rectangle +(5,5);
\draw (xi1) -- (xi11);
\draw (xi2) -- (xi21);
\draw (xi3) -- (xi31);
\draw (xi1) -- (xi10);
\draw (xi2) -- (xi20);
\draw (xi3) -- (xi30);
\end{tikzpicture}
\hspace{1cm}
\begin{tikzpicture}[every node/.style={scale=0.7},scale=0.6] 

\draw(7,0) node[fill=black,circle] (xi3) {};
\draw(7,1.5) node[fill=black,circle] (xi2) {};
\draw(7,3) node[fill=black,circle] (xi1){};
\draw (6.5,-0.5) rectangle +(1,4);

\draw(3.5,0) node[fill=black,circle] (xi31) {};
\draw(2,1.5) node[fill=black,circle] (xi21) {};
\draw(0.5,3) node[fill=black,circle] (xi11){};
\draw(2.5,-0.5) node[fill=black,circle] (xi30) {};
\draw(1,1) node[fill=black,circle] (xi20) {};
\draw(-0.5,2.5) node[fill=black,circle] (xi10){};
\draw (-1,-1) rectangle +(5,5);
\draw (xi1) -- (xi11);
\draw (xi2) -- (xi21);
\draw (xi3) -- (xi31);
\draw (xi1) -- (xi10);
\draw (xi2) -- (xi20);
\draw (xi3) -- (xi30);
\end{tikzpicture}}}

\vspace{0.25cm}
\caption{\gn{Cl} -- \gn{Var-not-in-Cl} and \gn{Var-not-in-Cl} -- \gn{Cl}}
\end{subfigure}

\begin{subfigure}[b]{\textwidth}
\centering
\vspace{0.25cm}
\scalebox{0.8}{\vbox{
\begin{tikzpicture}[every node/.style={scale=0.7},scale=0.6] 

\draw(3,4) node[fill=black,circle] (xi3) {};
\draw(1.5,4.5) node[fill=black,circle] (xi2) {};
\draw(0,5) node[fill=black,circle] (xi1){};

\draw(8,4) node[fill=black,circle] (xi31) {};
\draw(6.5,4.5) node[fill=black,circle] (xi21) {};
\draw(5,5) node[fill=black,circle] (xi11){};

\draw (-0.5,3.5) rectangle +(4,2);
\draw (4.5,3.5) rectangle +(4,2);

\draw (xi1) -- (xi11);
\draw (xi2) -- (xi21);
\draw (xi3) -- (xi31);

\draw(0,2.5) node {~};
\end{tikzpicture}
\hspace{0.25cm}
\begin{tikzpicture}[every node/.style={scale=0.7},scale=0.6] 
\draw(3.5,3) node[fill=black,circle] (xi31) {};
\draw(2,4.5) node[fill=black,circle] (xi21) {};
\draw(0.5,6) node[fill=black,circle] (xi11){};
\draw(2.5,2.5) node[fill=black,circle] (xi30) {};
\draw(1,4) node[fill=black,circle] (xi20) {};
\draw(-0.5,5.5) node[fill=black,circle] (xi10){};

\draw(9.5,3) node[fill=black,circle] (xi311) {};
\draw(8,4.5) node[fill=black,circle] (xi211) {};
\draw(6.5,6) node[fill=black,circle] (xi111){};
\draw(8.5,2.5) node[fill=black,circle] (xi301) {};
\draw(7,4) node[fill=black,circle] (xi201) {};
\draw(5.5,5.5) node[fill=black,circle] (xi101){};

\draw (-1,2) rectangle +(5,5);
\draw (5,2) rectangle +(5,5);

\draw (xi31) -- (xi311);
\draw (xi31) -- (xi301);
\draw (xi30) -- (xi311);
\draw (xi30) -- (xi301);

\draw (xi21) -- (xi211);
\draw (xi21) -- (xi201);
\draw (xi20) -- (xi211);
\draw (xi20) -- (xi201);

\draw (xi11) -- (xi111);
\draw (xi11) -- (xi101);
\draw (xi10) -- (xi111);
\draw (xi10) -- (xi101);
\end{tikzpicture} 
}}
\vspace{0.25cm}
\caption{\gn{Var-in-Cl} -- \gn{Var-in-Cl} and \gn{Var-not-in-Cl} -- \gn{Var-not-in-Cl}}
\end{subfigure}

\begin{subfigure}[b]{\textwidth}
\centering
\vspace{0.25cm}
\scalebox{0.8}{\vbox{
\begin{tikzpicture}[every node/.style={scale=0.7},scale=0.6] 
\draw(8,4) node[fill=black,circle] (xi3) {};
\draw(6.5,4.5) node[fill=black,circle] (xi2) {};
\draw(5,5) node[fill=black,circle] (xi1){};

\draw(3.5,3) node[fill=black,circle] (xi311) {};
\draw(2,4.5) node[fill=black,circle] (xi211) {};
\draw(0.5,6) node[fill=black,circle] (xi111){};
\draw(2.5,2.5) node[fill=black,circle] (xi301) {};
\draw(1,4) node[fill=black,circle] (xi201) {};
\draw(-0.5,5.5) node[fill=black,circle] (xi101){};

\draw (4.5,3.5) rectangle +(4,2);
\draw (-1,2) rectangle +(5,5);

\draw (xi3) -- (xi311);
\draw (xi3) -- (xi301);
\draw (xi2) -- (xi211);
\draw (xi2) -- (xi201);
\draw (xi1) -- (xi111);
\draw (xi1) -- (xi101);
\end{tikzpicture} 
\hspace{0.25cm}
\begin{tikzpicture}[every node/.style={scale=0.7},scale=0.6] 
\draw(3,4) node[fill=black,circle] (xi3) {};
\draw(1.5,4.5) node[fill=black,circle] (xi2) {};
\draw(0,5) node[fill=black,circle] (xi1){};

\draw(9.5,3) node[fill=black,circle] (xi311) {};
\draw(8,4.5) node[fill=black,circle] (xi211) {};
\draw(6.5,6) node[fill=black,circle] (xi111){};
\draw(8.5,2.5) node[fill=black,circle] (xi301) {};
\draw(7,4) node[fill=black,circle] (xi201) {};
\draw(5.5,5.5) node[fill=black,circle] (xi101){};

\draw (-0.5,3.5) rectangle +(4,2);
\draw (5,2) rectangle +(5,5);

\draw (xi3) -- (xi311);
\draw (xi3) -- (xi301);
\draw (xi2) -- (xi211);
\draw (xi2) -- (xi201);
\draw (xi1) -- (xi111);
\draw (xi1) -- (xi101);
\end{tikzpicture} }}
\vspace{0.25cm}
\caption{\gn{Var-not-in-Cl} -- \gn{Var-in-Cl} and \gn{Var-in-Cl} -- \gn{Var-not-in-Cl}}
\end{subfigure}

\end{center}
\caption{The horizontal connections among contiguous gadgets in a clause row.    \label{fig:grid-gadget-Horz}}
\end{figure}

\clearpage
\paragraph{The color classes}
We describe first the color classes of  the graph $(G,\vp)$.  Each color class corresponds to one of the basic gadgets.  An example of the construction is given in  Figure~\ref{fig:grid-example}. The cluster graph (described here as a grid of  gadgets)  has one column for each variable (in the order $x_1 \dots x_n$) and two additional columns, first and last. The upper row starts and ends with a Pad gadget and it has  one Var gadget in each of the columns, i.e., one  for each variable.

For each clause, we create two consecutive rows in the grid, following the order of  the clauses $C_1, \dots C_m$.  The upper row associated to a clause $C_j$, starts an ends with a \gn{Cl} gadget.  At column $i$, we place  a   \gn{Var-in-Cl} gadget, if $x_i\in C_j$,  or a \gn{Var-not-in-Cl}  gadget, otherwise. The lower row associated to $C_j$ starts and ends with a \gn{Pad} gadget and contains one \gn{Var} gadget for each variable.

\paragraph{The connections among vertices}
The connections among vertices in the different color classes  depends on the type of gadget and on whether the two color classes  are connected in the grid vertically or horizontally.  Let us start with the vertical connections.
The vertex in a \gn{Pad} gadget is connected to all the vertices in the vertically contiguous \gn{Cl} gadgets. The vertical connections of  a \gn{Var-in-Cl}  or  a \gn{Var-not-in-Cl} gadget and its upper  and lower \gn{Var} gadgets are given in Figure~\ref{fig:grid-gadget-Vert}.

The horizontal connections are the following:  The vertex in a \gn{Pad} gadget is connected to all the vertices in the horizontally contiguous \gn{Var} gadget. The vertices in two horizontally contiguous  \gn{Var} gadgets  are connected by a complete bipartite graph.  The other horizontal connections corresponds to  contiguous pairs of  gadgets from the  types   \gn{Cl},  \gn{Var-in-Cl} and \gn{Var-not-in-Cl}. The connections among all the possible combinations of such pairs are described in Figure~\ref{fig:grid-gadget-Horz}.


Note that, the vertical connections described in Figure~\ref{fig:grid-gadget-Vert} guarantee that the left (right) vertex in a \gn{Var} gadget is connected by a path only to a left (right) vertex in another \gn{Var} gadget in the same column.  Furthermore, the horizontal connections, as described in  \ref{fig:grid-gadget-Horz}, always join vertices in the same vertical position  (upper, middle or lower).  

\begin{figure}[t]
\scalebox{0.35}{\vbox{
\begin{center}
\begin{tikzpicture}[every node/.style={shape=circle,fill=black,scale=0.7}] 
\draw (-3,0.5) rectangle +(1,1);
\draw (-2.5,1) node  (c00) {};
\draw[red]  (-2.5,1) circle (0.3cm);
\draw (-3,3.5) rectangle +(1,2);
\draw (-2.5,5) node (c101) {};
\draw[red]  (-2.5,5) circle (0.3cm);
\draw (-2.5,4.5) node (c102) {}; 
\draw (-2.5,4) node (c103) {};

\draw (2,0.5) rectangle +(2,1);
\draw (2.5,1) node (c011) {};
\draw (3.5,1) node (c012) {}; 
\draw[red]  (3.5,1) circle (0.3cm);

\draw (1,3) rectangle +(5,3);
\draw (2,5) node (c111) {};
\draw (2.5,5) node (c112) {};
\draw[red]  (2.5,5) circle (0.3cm); 
\draw (3.5,4.5) node (c113) {};
\draw (4,4.5) node (c114) {};
\draw (5,4) node (c115) {};
\draw (5.5,4) node (c116) {};

\draw (9,0.5) rectangle +(2,1);
\draw (9.5,1) node (c021) {};
\draw[red]  (9.5,1) circle (0.3cm);
\draw (10.5,1) node (c022) {}; 

\draw (8.5,3.5) rectangle +(3,2);
\draw (9,5) node (c121) {};
\draw[red]  (9,5) circle (0.3cm);
\draw (10,4.5) node (c122) {}; 
\draw (11,4) node (c123) {};

\draw (15,0.5) rectangle +(2,1);
\draw (15.5,1) node (c031) {};
\draw (16.5,1) node (c032) {}; 
\draw[red]  (16.5,1) circle (0.3cm);

\draw (14.5,3.5) rectangle +(3,2);
\draw (15,5) node (c131) {};
\draw[red]  (15,5) circle (0.3cm);
\draw (16,4.5) node (c132) {}; 
\draw (17,4) node (c133) {};

\draw (22,0.5) rectangle +(2,1);
\draw (22.5,1) node (c041) {};
\draw (23.5,1) node (c042) {}; 
\draw[red]  (23.5,1) circle (0.3cm);

\draw (21.5,3.5) rectangle +(3,2);
\draw (22,5) node (c141) {};
\draw[red]  (22,5) circle (0.3cm);
\draw (23,4.5) node (c142) {}; 
\draw (24,4) node (c143) {};

\draw (28,0.5) rectangle +(1,1);
\draw (28.5,1) node  (c05) {};
\draw[red]  (28.5,1) circle (0.3cm);

\draw (28,3.5) rectangle +(1,2);
\draw (28.5,5) node (c151) {};
\draw[red]  (28.5,5) circle (0.3cm);
\draw (28.5,4.5) node (c152) {}; 
\draw (28.5,4) node (c153) {};


\draw (-3,8.5) rectangle +(1,1);
\draw (-2.5,9) node  (c20) {};
\draw[red]  (-2.5,9) circle (0.3cm);
\draw (-3,11.5) rectangle +(1,2);
\draw (-2.5,13) node (c301) {};
\draw (-2.5,12.5) node (c302) {}; 
\draw[red]  (-2.5,12.5) circle (0.3cm);
\draw (-2.5,12) node (c303) {};

\draw (2,8.5) rectangle +(2,1);
\draw (2.5,9) node (c211) {};
\draw (3.5,9) node (c212) {};
\draw[red]  (3.5,9) circle (0.3cm); 

\draw (1.5,11.5) rectangle +(3,2);
\draw (2,13) node (c311) {};
\draw (3,12.5) node (c312) {};
\draw[red]  (3,12.5) circle (0.3cm); 
\draw (4,12) node (c313) {};

\draw (9,8.5) rectangle +(2,1);
\draw (9.5,9) node (c221) {};
\draw[red]  (9.5,9) circle (0.3cm);
\draw (10.5,9) node (c222) {};

\draw (8.5,11.5) rectangle +(3,2);
\draw (9,13) node (c321) {};
\draw (10,12.5) node (c322) {}; 
\draw[red]  (10,12.5) circle (0.3cm);
\draw (11,12) node (c323) {};

\draw (15,8.5) rectangle +(2,1);
\draw (15.5,9) node (c231) {};
\draw (16.5,9) node (c232) {}; 
\draw[red]  (16.5,9) circle (0.3cm);

\draw (14.5,11.5) rectangle +(3,2);
\draw (15,13) node (c331) {};
\draw (16,12.5) node (c332) {};
\draw[red]  (16,12.5) circle (0.3cm); 
\draw[red]  (16,12.5) circle (0.3cm);
\draw (17,12) node (c333) {};

\draw (22,8.5) rectangle +(2,1);
\draw (22.5,9) node (c241) {};
\draw (23.5,9) node (c242) {}; 
\draw[red]  (23.5,9) circle (0.3cm);

\draw (20.5,11) rectangle +(5,3);
\draw (21.5,13) node (c341) {};
\draw (22,13) node (c342) {}; 
\draw (23,12.5) node (c343) {};
\draw (23.5,12.5) node (c344) {};
\draw[red]  (23.5,12.5) circle (0.3cm);
\draw (24.5,12) node (c345) {};
\draw (25,12) node (c346) {};

\draw (28,8.5) rectangle +(1,1);
\draw (28.5,9) node (c25) {};
\draw[red]  (28.5,9) circle (0.3cm);
\draw (28,11.5) rectangle +(1,2);
\draw (28.5,13) node (c351) {};
\draw (28.5,12.5) node (c352) {}; 
\draw[red]  (28.5,12.5) circle (0.3cm);
\draw (28.5,12) node (c353) {};

\draw (-3,16) rectangle +(1,1);
\draw (-2.5,16.5) node (c40) {};
\draw[red]  (-2.5,16.5) circle (0.3cm);

\draw (2,16) rectangle +(2,1);
\draw (2.5,16.5) node (c411) {};
\draw (3.5,16.5) node (c412) {}; 
\draw[red]  (3.5,16.5) circle (0.3cm);

\draw (9,16) rectangle +(2,1);
\draw (9.5,16.5) node (c421) {};
\draw[red]  (9.5,16.5) circle (0.3cm);
\draw (10.5,16.5) node (c422) {}; 

\draw (15,16) rectangle +(2,1);
\draw (15.5,16.5) node (c431) {};
\draw (16.5,16.5) node (c432) {}; 
\draw[red]  (16.5,16.5) circle (0.3cm);

\draw (22,16) rectangle +(2,1);
\draw (22.5,16.5) node (c441) {};
\draw (23.5,16.5) node (c442) {}; 
\draw[red]  (23.5,16.5) circle (0.3cm);

\draw (28,16) rectangle +(1,1);
\draw (28.5,16.5) node (c45) {};
\draw[red]  (28.5,16.5) circle (0.3cm);

\path[draw]  (c40) edge [bend left] (c411)  edge [bend left] (c412);
\path[draw] (c411) edge [bend left] (c421) edge [bend left] (c422);
\path[draw] (c412) edge [bend left] (c421) edge [bend left] (c422);
\path[draw] (c421) edge [bend left] (c431) edge [bend left] (c432);
\path[draw] (c422) edge [bend left] (c431) edge [bend left] (c432);
\path[draw] (c431) edge [bend left] (c441) edge [bend left] (c442);
\path[draw] (c432) edge [bend left] (c441) edge [bend left] (c442);
\path[draw] (c441) edge [bend left] (c45);
\path[draw] (c442) edge [bend left] (c45);

\path[draw] (c301) -- (c311) -- (c321) -- (c331) -- (c341);
\path[draw] (c302) -- (c312) -- (c322) -- (c332) -- (c343);
\path[draw] (c303) -- (c313) -- (c323) -- (c333) -- (c345);
\path[draw] (c342) -- (c351); 
\path[draw] (c344) -- (c352); 
\path[draw] (c346) -- (c353); 

\path[draw] (c331) edge [bend left] (c342); 
\path[draw] (c332) edge [bend left] (c344); 
\path[draw] (c346) edge [bend left] (c333); 

\path[draw] (c341) edge [bend left] (c351);
\path[draw] (c343) edge [bend left] (c352);
\path[draw] (c353) edge [bend left] (c345);
 
\path[draw] (c25) edge [bend left] (c241)  edge [bend left] (c242);
\path[draw] (c241) edge [bend left] (c231) edge [bend left] (c232);
\path[draw] (c242) edge [bend left] (c231) edge [bend left] (c232);
\path[draw] (c231) edge [bend left] (c221) edge [bend left] (c222);
\path[draw] (c232) edge [bend left] (c221) edge [bend left] (c222);
\path[draw] (c221) edge [bend left] (c211) edge [bend left] (c212);
\path[draw] (c222) edge [bend left] (c211) edge [bend left] (c212);
\path[draw] (c211) edge [bend left] (c20);
\path[draw] (c212) edge [bend left] (c20);

\path[draw] (c112) -- (c121) -- (c131) -- (c131) -- (c141) -- (c151);
\path[draw] (c114) -- (c122) -- (c132) -- (c132) -- (c142) -- (c152);
\path[draw] (c116) -- (c123) -- (c133) -- (c133) -- (c143) -- (c153);
\path[draw] (c101) -- (c111); 
\path[draw] (c102) -- (c113); 
\path[draw] (c103) -- (c115); 

\path[draw] (c101) edge [bend left] (c112); 
\path[draw] (c114) edge [bend left] (c102); 
\path[draw] (c116) edge [bend left] (c103); 

\path[draw] (c111) edge [bend left] (c121);
\path[draw] (c122) edge [bend left] (c113);
\path[draw] (c123) edge [bend left] (c115);

\path[draw] (c05) edge [bend left] (c041)  edge [bend left] (c042);
\path[draw] (c041) edge [bend left] (c031) edge [bend left] (c032);
\path[draw] (c042) edge [bend left] (c031) edge [bend left] (c032);
\path[draw] (c031) edge [bend left] (c021) edge [bend left] (c022);
\path[draw] (c032) edge [bend left] (c021) edge [bend left] (c022);
\path[draw] (c021) edge [bend left] (c011) edge [bend left] (c012);
\path[draw] (c022) edge [bend left] (c011) edge [bend left] (c012);
\path[draw] (c011) edge [bend left] (c00) ;
\path[draw] (c012) edge [bend left] (c00);


\path[draw] (c40) edge (c301) edge [bend left] (c302);
\path[draw] (c303) edge [bend left] (c40) ;
\path[draw] (c303) -- (c20) -- (c101);
\path[draw] (c302) edge [bend left] (c20);
\path[draw] (c20) edge [bend left] (c301);

\path[draw] (c20) edge (c101) edge [bend left] (c102);
\path[draw] (c103) edge [bend left] (c20);
\path[draw] (c103) -- (c00);
\path[draw] (c102) edge [bend left] (c00);
\path[draw] (c00) edge [bend left] (c101);

\path[draw] (c45) edge (c351) edge [bend left] (c352);
\path[draw] (c353) edge [bend left] (c45) ;
\path[draw] (c353) -- (c25) -- (c151);
\path[draw] (c352) edge [bend left] (c25);
\path[draw] (c25) edge [bend left] (c351);

\path[draw] (c25) edge (c151) edge [bend left] (c152);
\path[draw] (c153) edge [bend left] (c25);
\path[draw] (c153) -- (c05);
\path[draw] (c152) edge [bend left] (c05);
\path[draw] (c05) edge [bend left] (c151);

\path[draw] (c411) -- (c311) -- (c211);
\path[draw] (c412) -- (c312) -- (c212);
\path[draw] (c412) -- (c313) -- (c212);

\path[draw] (c421) -- (c322) -- (c221);
\path[draw] (c422) -- (c321) -- (c222);
\path[draw] (c422) -- (c323) -- (c222);

\path[draw] (c431) -- (c333) -- (c231);
\path[draw] (c432) -- (c331) -- (c232);
\path[draw] (c432) -- (c332) -- (c232);

\path[draw] (c441) -- (c341) -- (c241);
\path[draw] (c442) -- (c342) -- (c242);
\path[draw] (c441) -- (c343) -- (c241);
\path[draw] (c442) -- (c344) -- (c242);
\path[draw] (c441) -- (c345) -- (c241);
\path[draw] (c442) -- (c346) -- (c242);

\path[draw] (c211) -- (c111) -- (c011);
\path[draw] (c212) -- (c112) -- (c012);
\path[draw] (c211) -- (c113) -- (c011);
\path[draw] (c212) -- (c114) -- (c012);
\path[draw] (c211) -- (c115) -- (c011);
\path[draw] (c212) -- (c116) -- (c012);

\path[draw] (c221) -- (c121) -- (c021);
\path[draw] (c222) -- (c122) -- (c022);
\path[draw] (c222) -- (c123) -- (c022);

\path[draw] (c231) -- (c132) -- (c031);
\path[draw] (c232) -- (c131) -- (c032);
\path[draw] (c232) -- (c133) -- (c032);

\path[draw] (c241) -- (c143) -- (c041);
\path[draw] (c242) -- (c141) -- (c042);
\path[draw] (c242) -- (c142) -- (c042);

\end{tikzpicture} 
\end{center}}}
\caption{The colored graph obtained from the monotone formula $\Phi=((x_1,x_2,x_3),(x_2,x_3,x_4))$. The circled vertices form a multicolored realization of the cluster graph.  \label{fig:grid-example}}
\end{figure}
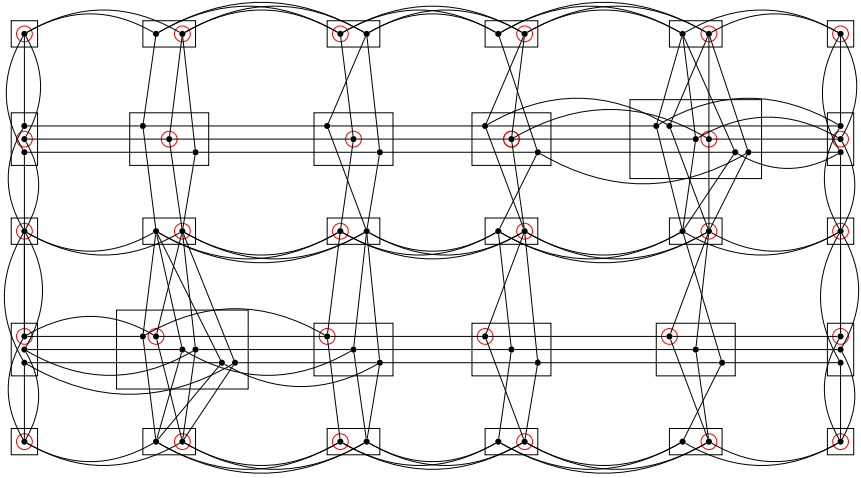

\medskip
\paragraph{Correctness of the reduction}
Observe that the graph together with the coloring  can be constructed in polynomial time. 
 
\remove{  Furthermore, by interpreting that the label of the selected vertex in the Var gadget associated to variable $x_i$ indicates the truth value assigned  to the variable,  we can show that the construction is indeed a reduction. }

Let us start proving that when $\Phi$ is a yes instance of the 1-in-3 Monotone 3-SAT problem, the constructed colored graph admits a multicolored realization. 
Let $T$ be a valid  assignment to $\Phi$, i.e., letting $T(x_i)=t_i\in \{0,1\}$,  for $1\leq i \leq n$,  exactly one variable in each clause in $\Phi$ is set to $1$.  

We define a set $S\subseteq V$ as follows.
\begin{itemize}
\item The unique vertex in any \gn{Pad} gadget belong to $S$.
\item For each  \gn{Var} gadget corresponding to  variable $x_i$ we add to $S$  the  left vertex, if $t_i=1$, or the   right vertex, if $t_i=0$. 
\item Consider a   clause $C_j=(x_{j_1},x_{j_2},x_{j_3})$.  
 \begin{itemize}
 \item From the \gn{Cl} and the \gn{Var-in-Cl} gadgets in the  column associated to $C_j$, we add to $S$  the upper vertex, if $t_{j_1}t_{j_2}t_{j_3}=100$, the middle one, if $t_{j_1}t_{j_2}t_{j_3}=010$,  or the lower one, if $t_{j_1}t_{j_2}t_{j_3}=001$.
 \item From a variable $x_i$ that does not appear in $C_j$, we select the upper, middle of lower block, depending on whether $t_{j_1}t_{j_2}t_{j_3}$ is 100, 010 or 001. Inside the selected block, we add to $S$ the left vertex if $x_i=1$ or  the right vertex if $x_i=0$ to be added to $S$.  
 \end{itemize} 
\end{itemize}
Note that,  $S$ contains exactly one vertex from each color class in $(G,\vp)$ so, it is multicolored.  It remains to  show that $S$ is a realization of the cluster graph. Let us first look to the horizontal connections.  The circled vertices in the colored graph given in Figure \ref{fig:grid-example} are the set $S$ associated to the assignment $x_1= 0$, $x_2=1$, $x_3=0$ and $x_4=0$.

For the rows that alternate \gn{Var} clauses the connections are all to all,  therefore  $S$ realizes all the corresponding connections in the cluster graph. The same happens for the first and the last columns.

Consider a clause  $C_j=(x_{j_1},x_{j_2},x_{j_3})$ and the  corresponding assigned values    $t_{j_1}t_{j_2}t_{j_3}$. The vertices in $S$ from the \gn{Cl}, \gn{Var-in-Cl} and \gn{Var-not-in-Cl} in the row are all in the same vertical position (upper, middle, or lower).  Therefore, according to the horizontal connections, the horizontal path on this row is  realized by $S$  (see Figs.~\ref{fig:grid-gadget-Horz} and \ref{fig:grid-example}). 
  
Consider a variable $x_i$ with assigned value $t_i$, observe that depending on the value $t_i$, the selected vertices are all on the left ($t_i=1$) or on the right ($t_i=0$).  Therefore, according to the vertical connections (see  Figs.~\ref{fig:grid-gadget-Vert} and \ref{fig:grid-example}), the vertical  path on this column  is  realized by $S$.  We conclude that $S$ is a multicolored realization for $G_\vp$

To show the opposite direction, that when  $G_\vp$  has a multicolored realization $S$, $\Phi$ has a truth assignment $T$ in which each clause gets exactly 1 variable with assigned value 1.    We define $T$ as follows, consider the \gn{Var} gadgets on the top row, we  set  $T(x_i)=t_i$, being $t_i=1$ when $S$ contains the left vertex in the gadget and $t_i=0$ otherwise.  As $S$ is multicolored,  each $x_i$ is assigned a single truth value.   In Figure \ref{fig:grid-example}, the multicolored set defined by the circled vertices translates to the truth assignment $x_1= 0$, $x_2=1$, $x_3=0$ and $x_4=0$.  

Consider a column corresponding to a variable $x_i$ assigned to value $t_i$.  When the vertex selected on the top row gadget is the left (right) one, the vertical connections only allow vertical paths  that go through left (right) vertices in the gadgets in the column (see Figure~\ref{fig:grid-gadget-Vert}).  Therefore, for the \gn{Var} gadgets in column $i$,  when $t_i=1$, $S$ contains all the left vertices, and,  when $t_i=0$,  $S$ contains  all the right vertices. 

Consider a clause  $C_j=(x_{j_1},x_{j_2},x_{j_3})$.  And consider the vertical position (upper, middle or lower) of the vertex in the leftmost \gn{Cl}  gadget in the corresponding row. Recall that, according to  the definition of the horizontal connections, a horizontal complete path in the cluster graph can only contain vertices in the same vertical position in all the gadgets.  
Consider the path $p$ induced by the vertices in $S$ from the clusters in the associated row starting from the left. If $p$ starts in the upper vertex, it contains the upper vertices of the \gn{Var-in-Cl} gadget associated to the variables $x_{j_1}$, $x_{j_2}$ and $x_{j_3}$. The first one is connected only to the left vertex on the vertically contiguous \gn{Var} gadgets, while the other two are connected only to the right vertex on the vertically contiguous \gn{Var} gadgets (see Figure~\ref{fig:grid-gadget-Vert}). So, we get that $t_{j_1}=1$,  $t_{j_2}=0$, and $t_{j_3}=0$. A similar argument shows that when $p$ starts in the middle (lower) vertex, then $t_{j_1}=0$,  $t_{j_2}=1$, and $t_{j_3}=0$ ($t_{j_1}=0$,  $t_{j_2}=0$, and $t_{j_3}=1$). Therefore,  the constructed formula is a yes instance of the 1-in-3 Monotone SAT problem.
\end{proof}

\section{Colored graphs with bounded cluster size}\label{sec:5}

We start presenting  a polynomial time algorithm for the the particular case of the \MGRP problem in which the colored graph has cluster size  at most 2.  Later,  we  show that the problem  becomes NP-complete for graphs with cluster size larger than 2,  thus providing a complexity dichotomy with respect to cluster size.  In order to get the result, we provide a reduction  to the 2-SAT problem: given a boolean formula $\Phi$ in CNF with at most two literals per clause, decide whether $\Phi$ has a satifying assignment. Recall that the 2-SAT problem can be  solved in polynomial time \cite{Garey79}.

\begin{proposition}\label{prop:size2}
 The {\MGRP} problem is polynomial time solvable for colored graphs with cluster size at most 2. 
\end{proposition}
\begin{proof}
Let $(G,\vp)$ be a colored graph with cluster size  $s(G,\vp)=2$.  Let $A_1, \dots,A_k$ be the color classes in  $(G,\vp)$. 

From $(G,\vp)$, we create an instance $\Phi$ of 2-SAT as  follows.  $\Phi$ has variables $x_1,\dots x_k$. We associate to each vertex in $u\in V$ a literal $l_u$.   
If $A_i= \{u\}$,  $l_u=x_i$.  If $A_i=\{u,v\}$,  $l_u = x_i$ and $l_v=\neg x_i$.
The clauses in $\Phi$ are the following. For each color class  with $|A_i|=1$, we add clause $x_i$  
For each edge  $(A_i,A_j)\in  E(G_\vp)$,  we add one clause for each missing edge among vertices in the color classes, i.e., for $u\in A_i$, $w\in A_j$ with $(u,w)\notin E(G)$, we  add the clause  $(\neg l_u \wedge  \neg l_w)$.  Observe that this clause will be satisfied only when at least one of its two literals is assigned value 0. 

An example of the construction is given in  Figure \ref{fig:2-sat}.
Note that  $\Phi $ is a 2-SAT instance and it can be constructed in polynomial time. 
Let us show  that   $(G,\vp)$ admits a multicolored realization if and only if $\Phi$ is satisfiable. 

Assume that  $(G,\vp)$ admits a multicolored realization $S$. Let us  consider  the assignment $T$ of  truth values to the variables of $\Phi$ that makes the literals associated to the vertices in $S$  get value 1, as  $S$ is multicolored, $T$ is a valid assignment for $\Phi$. 
As $S$ is multicolored it contains all the vertices in color classes with only one vertex. Therefore, $T$ satisfies all the clauses in $\Phi$ with one literal.

Consider a connection $(A_i,A_j)\in  E(G_\vp)$ with some associated clause in $\Phi$. For $u\in A_i$, $w\in A_j$ with $(u,w)\notin E(G)$, at least one of the vertices $u$ or $w$ cannot belong to $S$. Therefore at least one of $l_u$ or $l_w$ gets value 0 under $t_S$. Therefore, the clause  $(\neg l_u \wedge  \neg l_w)$ is satisfied by $T$. We conclude that $T$ satisfies $\Phi$

For the other direction, assume that $T$ is a satisfying assignment for $\Phi$.  Consider the set of vertices $S$ that contains those vertices $u\in V(G)$ such that $T(l_u)=1$. 
As $\alpha$ is an assignment, then $S$ is multicolored. 
Let $S= \{l_1,\dots l_k\}$.  Assume that $S$ is  not a multicolored realization of $(G,\vp)$. In such a case, there must be an edge $(A_i,A_j)\in  E(G_\vp)$ with $(l_i, l_j)\notin E(G)$. In such a case, the clause $(\neg l_i \wedge  \neg l_j)$  will not be satisfied, contradicting the fact that $\alpha$ is a satisfying assignment.   We conclude that $S$ is a multicolored realization of   $(G,\vp)$.
\end{proof}

\begin{figure}
\begin{center}\scalebox{0.8}{
\phantom{aa}
\begin{tikzpicture}
\tikzstyle{petit} = [draw, circle, fill= black,scale=0.5]
\draw (0,1.5) rectangle +(1,1);
\draw (0.5,2) node[petit]  (c0) {};
\draw (0.5,2.7) node[scale= 1] {$x_1$};
\draw[red]  (0.5,2) circle (0.3cm);
\draw (1,3.5) rectangle +(2,1);
\draw (1.5,4) node[petit] (c11) {};
\draw (1.5,4.7) node[scale= 1] {$x_2$};
\draw[red]  (1.5,4) circle (0.3cm);
\draw (2.5,4) node[petit] (c12) {}; 
\draw (2.5,4.7) node[scale= 1] {$\neg x_2$};
\draw (4,3.5) rectangle +(1,1);
\draw (4.5,4) node[petit]  (c2) {};
\draw (4.5,4.7) node[scale= 1] {$x_5$};
\draw[red]  (4.5,4) circle (0.3cm);
\draw (6,1.5) rectangle +(2,1);
\draw (6.5,2) node[petit]  (c31) {};
\draw (6.5,2.7) node[scale= 1] {$x_4$};
\draw[red]  (6.5,2) circle (0.3cm);
\draw (7.5,2) node[petit]  (c32) {};
\draw (7.5,2.7) node[scale= 1] {$\neg x_4$};
\draw (2,-1.5) rectangle +(2,1);
\draw (2.5,-1) node[petit] (c41) {};
\draw (2.5,-1.8) node[scale= 1] {$x_3$};
\draw[red]  (2.5,-1) circle (0.3cm);
\draw (3.5,-1) node[petit] (c42) {}; 
\draw (3.5,-1.8) node[scale= 1] {$\neg x_3$};

\path[draw] (c0) -- (c11) -- (c31) -- (c2);
\path[draw] (c11) -- (c41) -- (c31);
\path[draw] (c12) -- (c42);
\path[draw] (c41) -- (c32);
\path[draw] (c42) -- (c31);

\draw (12,2) node[scale = 1]{\vbox{\begin{align*}\Phi  = &x_1 \wedge x_5 \wedge (\neg x_1 \vee x_2)\\ & \wedge (\neg x_2 \vee x_3)  \wedge (x_2 \vee \neg x_3)\\&\wedge (x_2 \vee x_4)\wedge (x_2 \vee \neg x_4)\\&\wedge (\neg x_3 \vee x_4)\wedge (\neg x_5 \vee x_1)\end{align*}}};
\end{tikzpicture}}
\end{center}
\caption{A colored graph with cluster size at most two and the associated  2-SAT formula. The set formed by the circled vertices is a multicolored realization corresponding to the satisfying assignment $x_i=1$, for $i=1,\dots,5$.\label{fig:2-sat}}
\end{figure}
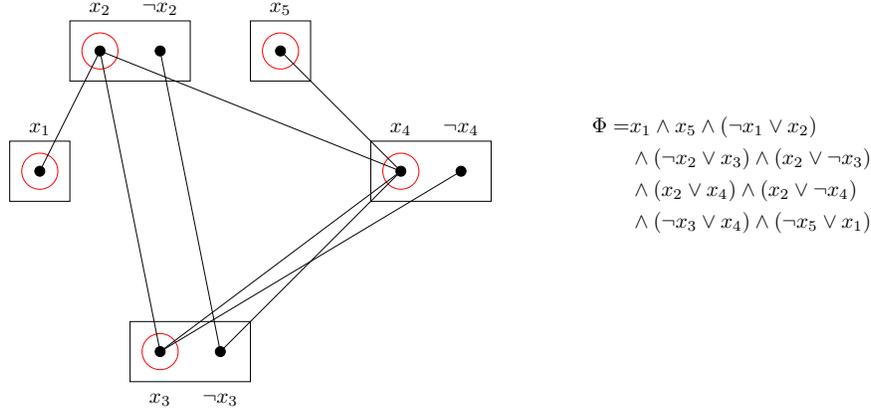

Let us analyze the case of colored graphs with cluster size $s\geq 3$. The reduction provided in the proof of Theorem \ref{thm:ConSke-NPhard} shows NP-hardness for $s=3$. To extend the reduction to a value of $s>3$, we just add to the graph constructed in this reduction a large enough set of independent vertices. Those independent vertices are colored in such a way that each color class is completed to have $s$ vertices. After the  addition of the independent vertices the cluster  graph remains the same. Furthermore, none of the added vertices can form part of a multicolored realization. Therefore, the problem is  NP-complete for $s> 3$.  Putting this together with Proposition~\ref{prop:size2}, we get a complexity dichotomy with respect to cluster size.

\begin{theorem}\label{teo:dicot}
The \MGRP problem is  NP-complete for colored graphs with cluster size $s\geq 3$, and polynomial time solvable otherwise. 
\end{theorem}


Our last result is an FPT algorithm  for the \MGRP problem parametrized by the treewidth of the cluster graph  and the cluster size. Recall that we have already established that the \MGRP problem parameterized by the treewith of the cluster graph is $W[1]$-hard (see Theorem~\ref{thm:btw}) and that it is NP-complete, for $s>2$, when the cluster graph is convex bipartite. Recall that convex bibartite graph can have unbounded treewidth.  Our algorithm uses dynamic programming on the tree decomposition.  

To describe the algorithm, we assume that as usual, together with the input $(G,\varphi)$,  we are given a nice tree decomposition $(T,X,r)$ of the cluster graph $G_\varphi$.  To simplify the explanation, we change slightly the notation. For a node $v\in T$, we consider two associated graphs $G_\varphi^u$, the subgraph induced in $G_\varphi$ by the union of all the  bags in the subtree rooted at $u$ (a subgraph of $G_\varphi$), and $G_u$, the subgraph induced in $G$ by the union of all the color classes appearing in a bag in the subtree rooted at $u$  (a subgraph of $G$). Observe that by definition $(G_u)_\varphi=G_\varphi^u $.  

To describe the elements in a bag, as each element in a bag corresponds to a color class, we refer directly to the color class.  

Given a collection of color classes $\mathcal{A}=\{A_1, \dots , A_k\}$, and a set $S\subseteq V$, we say that $S$ is \emph{multicolored with respect to}  $\mathcal{A}$ when $S$ contains exactly one vertex from each color class in $\mathcal{A}$.

\begin{theorem}\label{thm:Btw-FPT}
The  {\MGRP} problem when parameterized by the treewidth of  the cluster graph and the cluster size belongs to FPT. 
\end{theorem}
\begin{proof}
Let $(G,\vp)$ be a colored graph with cluster size $s$, and 
let $(T,X,r)$ be a a nice tree decomposition of the cluster graph $G_\varphi$ with width $w$.  The dynamic programming  algorithm will fill, for each node $v\in V(T)$, a boolean table
${M_v(S)}$ having  an entry for each multicolored subset $S$ with respect to the color classes included in  $X_v$. At the end of the algorithm,  ${M_v(S)}=1$ if there is a multicolored set $S'$ in $G_v$ realizing $G_\varphi^v$ such that $S\subseteq S'$, and otherwise ${M_v(S)}=0$.  
If so,  the value of $M_r(\emptyset)$ will determine whether there is a multicolored realization of $G$ or not. 
We deal with the table computation for each type of node in the nice tree decomposition separately, as each type of node requires a different kind of recursion and a different correctness guarantee.

\paragraph{Start node}
Let $u$ be start node of $T$, that is  a leaf,  with $X_u = \{A\}$ for some color class $A$ in $G_\varphi$. The multicolored subsets are formed by just one vertex in $C$. We set 
$M_u(\{x\}) = 1$, for $x\in A$,  and $M_u(\emptyset) = 0$.
As the graph $G_\varphi^u$ is an isolated vertex, the computed values are correct.

\paragraph{Introduce node}
Let $u$ be an {\em introduce} node of $T$, let $v$ be its unique  child and assume that  $A$ is the unique color class in $X_{u}-X_{v}$.

Then, for each $x\in A$ and each multicolored set $S$ with respect to  $X_v$,  if $G[\{x\}\cup S]$ is a realization of $G_\varphi^u[X_v]$ and  $M_v(S)=1$, we set $M_v(\{x\}\cup S)=1$, otherwise we set the value to 0.  

Note that all the multicolored sets with respect to $X_u$ are formed by a vertex in $A$ and a multicolored set $S$ with respect to $X_v$. Furthermore, $G_\varphi^v$ does not include the vertex  $X$.  For a multicolored set $S'$ of   $G_\varphi^u$ let $x\in A\cup S'$ and let $S$ be formed by the vertices in $S'$ belonging to the color classes in $X_v$.  Then $G[S']$ is  a realization of $G_\varphi^u$ if and only if   $G[\{x\}\cup S]$ is a realization of $G_\varphi[X_u]$ and $M_v(S)=1$.

\paragraph{Forget node}
Let $u$ be a {\em forget} node of $T$,  let  $v$ be  its 
unique child and assume that,  according to the definition, $C$ is the unique color class in  $X_{v}-X_{u}$.
Then, for  each $x\in A$ and each multicolored subsets $S$ with respect to $X_v$, we define $M_u(S) = \wedge_{x\in A} M_v (S\cup \{x\})\}$.
This expression provides the  correct value, as we are considering  all the possible multicolored supersets of $S$ with respect to $X_v$, if one of them is extendable to a multicolored realization, then $S$ is also extendable.

\paragraph{Join node}
Let $u$ be a join node of $T$ with children $v$ and $w$. In this case we have that $X_u=X_v= X_w$

Then,  for each multicolored subsets $S$ with respect to  $ X_u$, we set  $M_u(S) = M_u (S) \wedge M_w(S)$. This formula provides the correct value, as for the set $S$ to be extendable to a multicolored realization in $G_u$, $S$ must be extendable to a multicolored realization in  both $G_v$ and $G_w$.

\paragraph{Complexity}
The size of the tables associated to a node is upperbounded  by $s^w$, as we have to select on vertex from each color class with at most $s$ vertices  and the number of classes in a bag is at most $w$. 
To compute the entries the most complex operation is  a forget node in which we have to look at all the elements in a color class. This is number is upperbounded by $n$.  
As the total number of nodes is polynomial in the number of color classes, the total cost is $O(s^w p(n))$.  This function shows that the problem is fixed parameter tractable.  
\end{proof}

\section{Conclusions and further results}\label{sec:6}

We have introduced a new generalized  graph problem in order to assess the viability of the  associated cluster graph in terms of a possible existing realization. We have studied the complexity of the problem in the parameterized framework. Our results shed light on  the hardness of the  problem with respect to several parameters.

We can consider also a variant of the \MGRP problem in which, instead of asking for a multicolored realization of the cluster graph, we are interested in a multicolored realization of a given spanning subgraph:
\begin{quote}
\emph{Multicolored subgraph realization} problem (\MsGRP)\\
Given a graph $G=(V,E)$ together with a coloring $\varphi$, does there exists a  multicolored realization of $H=(V(G_\vp), E')$, for $E'\subseteq E(G_\vp)$?  
\end{quote}
This version of the problem captures another well known problem the multicolored independent set problem which is also known to be $W[1]$ hard parameterized by the number of colors (see for example \cite{CyganMbook2015} ). 
Observe that the \MsGRP problems includes, as a particular case, the \MGRP problem. Therefore, all the hardness results provided in this paper hold for the \MsGRP problem. 

Note that in the \MGRP problem, when an edge is not present in $G_\vp$,  none of the vertices in the corresponding color classes are connected. This is not always the case in the \MsGRP problem, an edge that is not present in the target graph $H$ might appear  in $G_\vp$.  When  $(A,B)\in E(G_\vp)$ but $(A,B)\notin E(H)$, a realization of $H$ must select two not connected vertices, one from $A$ and another from $B$. So, a necessary condition for the existence of  a multicolored realization of $H$ is that the bipartite graph $G[A,B]$ connecting the vertices in $A$ with the vertices in $B$ is not a complete bipartite graph.  If $G[A,B]\equiv K_{|A|,|B|}$ and $H$ does not contain the edge $(A,B)$,  we know that no multicolored realization exists. Otherwise,  we can assume that,  for each  $(A,B)\in E(G_\vp)$ with $(A,B)\notin E(H)$, we have $E(A,B)=\{(u,v)\in V(G)\mid u\in A, v\in B\}\subsetneq A\times B$. Under this assumption, we can consider the graph $G'$ in which, for each  $(A,B)\in E(G_\vp)$ with $(A,B)\notin E(H)$, we remove from $E(G)$ the edges in $E(A,B)$ and add the edges $A\times B \setminus E(A,B)$. Maintaining the same coloring, we have that $G'_\vp = G_\vp$ and that $H$ has a multicolored realization if and only if there is a multicolored realization of  $G'_\vp$.  Furthermore, $G'_\vp$ can be constructed in polynomial time in the size of $G$. In this way we obtain a polynomial time reduction from the \MsGRP problem to the \MGRP that preserves all the parameters considered in this paper. In consequence, all the positive results (polynomial time or FPT algorithms) devised for the \MGRP problem also hold for the \MsGRP problem.

Recall that a homomorphism from a graph $G = (V, E)$ to a graph $H = (V', F)$ is a function $f$ from $X$ to $Y$ such that, for each $(x,y)\in E$, $(f(x),f(y))\in F$. If $S\subseteq V$ and  $f$ is a homomorphism from $G$ to $G[S]$,  $f$ is a
\emph{retraction with respect to $S$}  if, for $x\in S$, $f(x)= x$ \cite{Quilliot1985}.  Inspired in this notion we consider the following problem.

\begin{quote}
\emph{Composed  retraction} problem (\CR)\\
Given  graphs $G=(V,E)$  and  $H = (V', F)$ together with a homomorphism  $f$ from $G$ to $H$, is there a homomorphism $g$ from $H$ to $G$ such that, for $x\in V'$,    $f(g(x))=x$?  
\end{quote}
Note that, for colored graphs $(G,\vp)$ in which the coloring is proper, $\vp$ is a homomorphism from  $G$ to $G_\vp$.  Furthermore, if $S$ is a multicolored realization of $G_\vp$, then the function $g$ assigning to each vertex in $G_\vp$ the corresponding colored vertex in $S$, is a homomorphism from $G_\vp$ to $G$ that verifies $\vp(g(x))=x$.
On the other  hand, if there is  a homomorphism $g$ from $G_\vp$ to $G$ such that, for $x\in V(G_\vp)$,    $\vp(g(x))=x$, then $g(V(G_\vp))$ is a multicolored realization of $G_\vp$. Therefore, the \MGRP problem, when the coloring is proper, is a subproblem of the \CR problem. Taking into account that  the coloring used in the reductions in this paper are proper, all the hardness results provided in this paper hold for the \CR problem.  It remains open to find other parameterizatins under which the \CR problem becomes tractable. 

\section*{Acknowledgments}

J. D\'{\i}az and  M. Serna are partially supported by funds from MINECO and EU FEDER under 
  grant TIN 2017-86727-C2-1-R AGAUR project ALBCOM 2017-SGR-786. 
\"O. Y. Diner is partially supported by the Scientific and Technological Research Council T\"ubitak under project BIDEB 2219-1059B191802095 and by Kadir Has University under project 2018-BAP-08. O. Serra is supported by the Spanish Ministry of Science under project MTM2017-82166-P.


\begin{thebibliography}{10}
\expandafter\ifx\csname url\endcsname\relax
  \def\url#1{\texttt{#1}}\fi
\expandafter\ifx\csname urlprefix\endcsname\relax\def\urlprefix{URL }\fi
\expandafter\ifx\csname href\endcsname\relax
  \def\href#1#2{#2} \def\path#1{#1}\fi

\bibitem{Vepstas17}
L.~Vep\u{s}tas, Graph quotiens: a topological approach to graphs, Bulletin of
  the Novosibirsk Computing Center, Computer Science 41 (2017) 55--89.

\bibitem{Fisch95}
M.~Fischetti, J.~J. Salazar~Gonz\'{a}lez, P.~Toth, The symmetric generalized
  traveling salesman polytope, Networks 26~(2) (1995) 113--123.

\bibitem{Fisch97}
M.~Fischetti, J.~J. Salazar~Gonz\'{a}lez, P.~Toth, A branch-and-cut algorithm
  for the symmetric generalized traveling salesman problem, Oper. Res. 45~(3)
  (1997) 378--394.
\newblock \href {https://doi.org/10.1287/opre.45.3.378}
  {\path{doi:10.1287/opre.45.3.378}}.

\bibitem{Fisch02}
M.~Fischetti, J.~Gonz\'{a}lez, P.~Toth, The generalized traveling salesman and
  orienteering problems, Kluwer, Dordrecht, 2002.

\bibitem{GhianiI2000}
G.~Ghiani, G.~Improta, An efficient transformation of the generalized vehicle
  routing problem, European Journal of Operational Research 122~(1) (2000)
  11--17.

\bibitem{PopMSP2018}
P.~Pop, O.~Matei, C.~Sabo, A.~Petrovan, A two-level solution approach for
  solving the generalized minimum spanning tree problem, European Journal of
  Operational Research 265~(2) (2018) 478--487.

\bibitem{DemangeMPR2014}
M.~Demange, J.~Monnot, P.~Pop, B.~Ries, On the complexity of the selective
  graph coloring problem in some special classes of graphs, Theoretical
  Computer Science 540-541 (2014) 82--102.

\bibitem{DemangeERT2015}
M.~Demange, T.~Ekim, B.~Ries, C.~Tanasescu, On some applications of the
  selective graph coloring problem, European Journal of Operational Research
  240 (2015) 307--314.

\bibitem{DrorHC2000}
M.~Dror, H.~Haouari, J.~Chaouachi, Generalized spanning trees, European Journal
  of Operations Research 120 (2000) 583--592.

\bibitem{FeremansLL2003}
C.~Feremans, M.~Labbe, G.~Laporte, Generalized network design problems,
  European Journal of Operations Research 148~(1) (2003) 1--13.

\bibitem{Pop12}
P.~Pop, Generalized Network Design Problems: Modeling and Optimization, De
  Gruyter, Germany, 2012.
\newblock \href {https://doi.org/10.1515/9783110267686}
  {\path{doi:10.1515/9783110267686}}.

\bibitem{pietrzak2003}
K.~Pietrzak, On the parameterized complexity of the fixed alphabet shortest
  common supersequence and longest common subsequence problems, Journal of
  Computer and System Sciences 67~(4) (2003) 757--771, parameterized
  Computation and Complexity 2003.
\newblock \href {https://doi.org/10.1016/S0022-0000(03)00078-3}
  {\path{doi:10.1016/S0022-0000(03)00078-3}}.

\bibitem{FellowsTCS2009}
M.~R. Fellows, D.~Hermelin, F.~Rosamond, S.~Vialette, On the parameterized
  complexity of multiple-interval graph problems, Theoretical Computer Science
  410~(1) (2009) 53--61.
\newblock \href {https://doi.org/10.1016/j.tcs.2008.09.065}
  {\path{doi:10.1016/j.tcs.2008.09.065}}.

\bibitem{Gavril1972}
F.~Gavril, Algorithms for minimum coloring, maximum clique, minimum covering by
  cliques and maximum independent set of a chordal graph, {SIAM} Journal on
  Computing 1 (1972) 180--187.

\bibitem{enright2014}
J.~Enright, T.~Stewart, G.~Tardos, On list coloring and list homomorphism of
  permutation and interval graphs, {SIAM} Journal on Discrete Mathematics
  28~(4) (2014) 1675--1685.

\bibitem{HuangJP2015}
S.~Huang, J.~H., D.~Paulusma, Narrowing the complexity gap for coloring
  $(c_s,p_t )$-free graphs, The Computer Journal~(11) (2015) 3074--3088.

\bibitem{DiazDSS21}
J.~D\'{i}az, O.~Y. Diner, M.~Serna, O.~Serra, On list $k$-coloring convex
  bipartite graphs, in: C.~Gentile, G.~Stecca, P.~Ventura (Eds.), Graphs and
  Combinatorial Optimization: from Theory to Applications CTW2020 Proceedings,
  Vol.~5 of Airo, Sringer, 2021, pp. 15--26.

\bibitem{diestel}
R.~Diestel, Graph Theory, Vol. 173 of Graduate Texts in Mathematics,
  Heidelberg: Springer-Verlag, 2017.

\bibitem{nussbaum2010}
D.~Nussbaum, S.~Pu, J.~Sack, T.~Uno, H.~Zarrabi-Zadeh, Finding maximum edge
  bicliques in convex bipartite graphs, Algorithmica 64~(2) (2010) 140--149.

\bibitem{CyganMbook2015}
M.~Cygan, F.~V. Fomin, L.~Kowalik, D.~Lokshtanov, D.~Marx, M.~Pilipczuk,
  M.~Pilipczuk, S.~Saurabh, Parameterized Algorithms, 1st Edition, Springer
  Publishing Company, Incorporated, 2015.

\bibitem{FlumG2006}
J.~Flum, M.~Grohe, Parameterized Complexity Theory, Texts in Theoretical
  Computer Science. An EATCS Series, Springer-Verlag Berlin Heidelberg, 2006.

\bibitem{DowneyF95-I}
R.~G. Downey, M.~R. Fellows, Fixed-parameter tractability and completeness.
  {I}. {B}asic results., SIAM Journal on Computing 24~(4) (1995) 873--921.

\bibitem{DowneyF95-II}
R.~G. Downey, M.~R. Fellows, Fixed-parameter tractability and completeness
  {II}: {O}n completeness for ${W}[1]$., Theoretical Computer Science
  141~(1--2) (1995) 109--131.

\bibitem{schaefer1978}
T.~Schaefer, The complexity of satisfiability problems, in: Proceedings of the
  10th Annual ACM Symposium on Theory of Computing, 1978, pp. 216--226.

\bibitem{spinrad1987}
J.~P. Spinrad, A.~Brandst\"{a}dt, L.~Stewart, Bipartite permutation graphs,
  Discrete Applied Mathematics 18 (1987) 279--292.

\bibitem{Garey79}
M.~R. Garey, D.~S. Johnson, Computers and intractability. {A} guide to the
  theory of {NP}-completeness, Freeman and Company, 1979.

\bibitem{Quilliot1985}
A.~Quilliot, A retraction problem in graph theory, Discrete Mathematics 54
  (1985) 61--72.

\end{thebibliography}

\end{document}